\documentclass[12pt,onecolumn,twoside]{IEEEtran}

\usepackage{calc}

\usepackage{multirow}
\usepackage{cite}
\usepackage{graphicx,subfigure}
\usepackage{psfrag}
\usepackage{amsmath,amssymb,amsthm}
\usepackage{color}
\usepackage{tikz} 

\usepackage{bbm}

\usepackage{cases}

\DeclareMathOperator*{\defeq}{\triangleq}

\newtheorem{theorem}{Theorem}
\newtheorem{corollary}{Corollary}

\newtheorem{lemma}{Lemma}

\newcommand{\bit}{\begin{itemize}}
\newcommand{\eit}{\end{itemize}}

\newcommand{\bc}{\begin{center}}
\newcommand{\ec}{\end{center}}

\newcommand{\ba}{\begin{array}}
\newcommand{\ea}{\end{array}}

\newcommand{\beq}{\begin{equation}}
\newcommand{\eeq}{\end{equation}}

\newcommand{\beqn}{\begin{equation*}}
\newcommand{\eeqn}{\end{equation*}}

\newcommand{\bean}{\begin{eqnarray*}}
\newcommand{\eean}{\end{eqnarray*}}
\newcommand{\bea}{\begin{eqnarray}}
\newcommand{\eea}{\end{eqnarray}}

\def\E{\mathbb{E}}

\newcommand{\Lc}{{\mathcal L}}

\newcommand{\Rc}{{\mathcal R}}
\newcommand{\Sc}{{\mathcal S}}

\newcommand{\Zc}{{\mathcal Z}}

\newtheorem{remark}{Remark}

\newcommand{\non}{\nonumber}

\newcommand{\Hen}{\mathbb{H}}
\newcommand{\hen}{\mathrm{h}}
\newcommand{\Imu}{\mathbb{I}}

\newcommand{\bln}{n}

\newcommand{\Ho}{\mathcal{H}_{\text{out}}}

\newcommand{\Hop}{\mathcal{H}'_{\text{out}}}

\IEEEoverridecommandlockouts

\begin{document}
\sloppy

\title{Optimal Secure GDoF of Symmetric Gaussian  Wiretap Channel with a Helper}
\author{Jinyuan Chen and  Chunhua Geng
\thanks{Jinyuan Chen is with Louisiana Tech University, Department of Electrical Engineering, Ruston, USA (email: jinyuan@latech.edu). Chunhua Geng is with Nokia Bell Labs, Murray Hill, NJ 07974 USA (email: chunhua.geng@nokia-bell-labs.com).  
This work was presented in part at the 2019 IEEE International Symposium on Information Theory. 
The work of Jinyuan Chen was partly supported by Louisiana Board of Regents Support Fund (BoRSF) Research Competitiveness Subprogram (RCS) under grant 32-4121-40336.}  
}

\maketitle
\thispagestyle{empty}

\begin{abstract}

We study a symmetric Gaussian wiretap channel with a helper, where a confidential message is sent from  a transmitter  to a legitimate receiver, in the presence of a helper and an eavesdropper, under a weak notion of secrecy constraint.  
For this setting,  we characterize the optimal secure generalized degrees-of-freedom (GDoF). The result reveals that, adding a helper can significantly increase the secure GDoF of the wiretap channel.  The result is supported by a new converse and a new scheme. 
In the proposed scheme, the helper sends a cooperative jamming signal at a specific power level and direction. In this way, it minimizes the penalty in GDoF incurred by the secrecy constraint.  
In the secure rate analysis, the techniques of noise removal and signal separation  are used.

\end{abstract}

\section{Introduction}

The study of information-theoretic secrecy dates back to Shannon's  work  of \cite{Shannon:49} in 1949. Since then, information-theoretic secrecy has been investigated  in varying  communication channels,  e.g., the wiretap channels \cite{Wyner:75, CsiszarKorner:78, CH:78}, multiple access channels with confidential messages and  wiretap multiple access channels  \cite{TY:08cj, TekinYener:08d, LP:08, BMK:10, LLP:11,  HKY:13, KG:15, BSP15}, the broadcast  channels with confidential messages \cite{LMSY:08, KTW:08, XCC:09, LLPS:10, ChiaGamal:12, BMK:13},   and the interference channels with confidential messages \cite{LMSY:08,LBPSV:09,LYT:08, YTL:08, HY:09, PDT:09,  KGLP:11, XU:14,  MDHS:14, MM:14o, XU:15,  GTJ:15, GJ:15, MM:16,  allerton:16, MU:16, FW:16,   ChenSecurityIT:17, MXU:17, ChenAllerton:18, ChenIC:18, ChenLiArxiv:18}. 
In those settings, the messages are transmitted over the channels with secrecy constraints, which often incur a  penalty in capacity (cf.~\cite{LMSY:08, LBPSV:09,  KGLP:11, HKY:13, XU:14, XU:15, GTJ:15, MM:16, MU:16, ChenAllerton:18,  ChenIC:18, LK:18}).   
One way to minimize the capacity penalty incurred by secrecy constraints is to add helper(s) into the channels  (see, e.g., \cite{XU:12, XU:14, XU:13, NY:13, NY:14, MXU:17, MU:16, FW:16, ChenLiArxiv:18, Tang+:11}). Specifically, the work in \cite{ChenLiArxiv:18} recently showed that adding a helper can \emph{totally} remove the  penalty in sum generalized degrees-of-freedom (GDoF), in a two-user symmetric Gaussian interference channel. 

In this work, we consider secure communications over a Gaussian wiretap channel with a helper. In this setting, a confidential message sent from  a transmitter  to a legitimate receiver needs to be secure from an eavesdropper,  in the presence of a helper.  
The wiretap channel and its variations have been considered as the basic channels for the investigation of information-theoretic secrecy.   
For example,  the wiretap channel with a helper can be considered as a specific case of an interference channel with only one confidential message. 
The insights gained from the former  can be very helpful in understanding the fundamental limits of the latter. 
For the Gaussian wiretap channel with one helper, the work in \cite{Tang+:11} provided both inner and outer bounds on the achievable secrecy rate at finite SNRs, where the achievability is based on unstructured Gaussian random codes and the derived inner and outer bounds do not match in general. 
In the wiretap channel with $M$ helpers, the works in \cite{XU:12, XU:14} showed that the secure degrees-of-freedom (DoF) is $\frac{M}{M+1}$ for almost all channel gains. The result is derived under the assumption that perfect channel state information (CSI) is available at the transmitters.  The work in \cite{XU:13, MXU:17} then showed that the same secure DoF  of $\frac{M}{M+1}$ is still achievable when the eavesdropper CSI is not available at the transmitters.  Another work in \cite{NY:13} studied a Gaussian wiretap channel with a helper, where a single antenna is equipped at  each of the transmitter and the legitimate receiver, while multiple antennas are equipped at each of the helper and the eavesdropper.  The result in \cite{NY:13} revealed that  the secure DoF $1/2$ is achievable irrespective of the number of antennas at the eavesdropper, as long as the number of antennas at the helper  is the same as the number of antennas at the  eavesdropper.  
In the setting of wiretap channel with a helper,  the previous DoF results were generalized  to the multiple-antenna scenario where multiple antennas are equipped at each node \cite{NY:14, MU:16}. 
The work in \cite{NY:14} used the assumption that perfect CSI is  available at the transmitters, while    the work in \cite{MU:16} used the assumption that the eavesdropper CSI is not available at the transmitters.
In all of those previous works in \cite{XU:12, XU:14, XU:13, MXU:17, NY:13, NY:14, MU:16} (except for \cite{Tang+:11}, which studied the secrecy rate at finite SNRs), the authors considered the secure DoF performance of the channels.  
The DoF metric is a form of capacity approximation. Under the DoF metric,   all the non-zero channel gains are treated equally strong, at the regime of high signal-to-noise ratio (SNR). However, in the communication channels the capacity is usually affected by   different channel strengths of different links.  Therefore, it motivates us to go beyond the DoF metric and consider a better form of  capacity approximation.  
GDoF metric is a generalization of DoF, which is able to capture the  capacity behavior when  different links have different channel strengths and  is very helpful in understanding  the capacity to within a constant gap (cf.~\cite{ETW:08}).
The work in \cite{FW:16} studied  the secure GDoF and secure capacity  of the Gaussian wiretap channel with a helper, as well as the Gaussian multiple access wiretap channel, where channel gain from the first transmitter to the eavesdropper is the same as the channel gain from the second transmitter to the eavesdropper. Note that, the setting considered in \cite{FW:16} has  symmetric channel gains at the wiretapper, which is different from our setting considered in this work.  
Also note that, the secure GDoF upper bound and the lower bound provided in  \cite{FW:16} are not matched for a certain range of channel parameters. 
In this work, we seek to characterize  the \emph{optimal} secure GDoF of a wiretap channel with a helper.

\begin{figure}[t!]
\centering
\includegraphics[width=8.9cm]{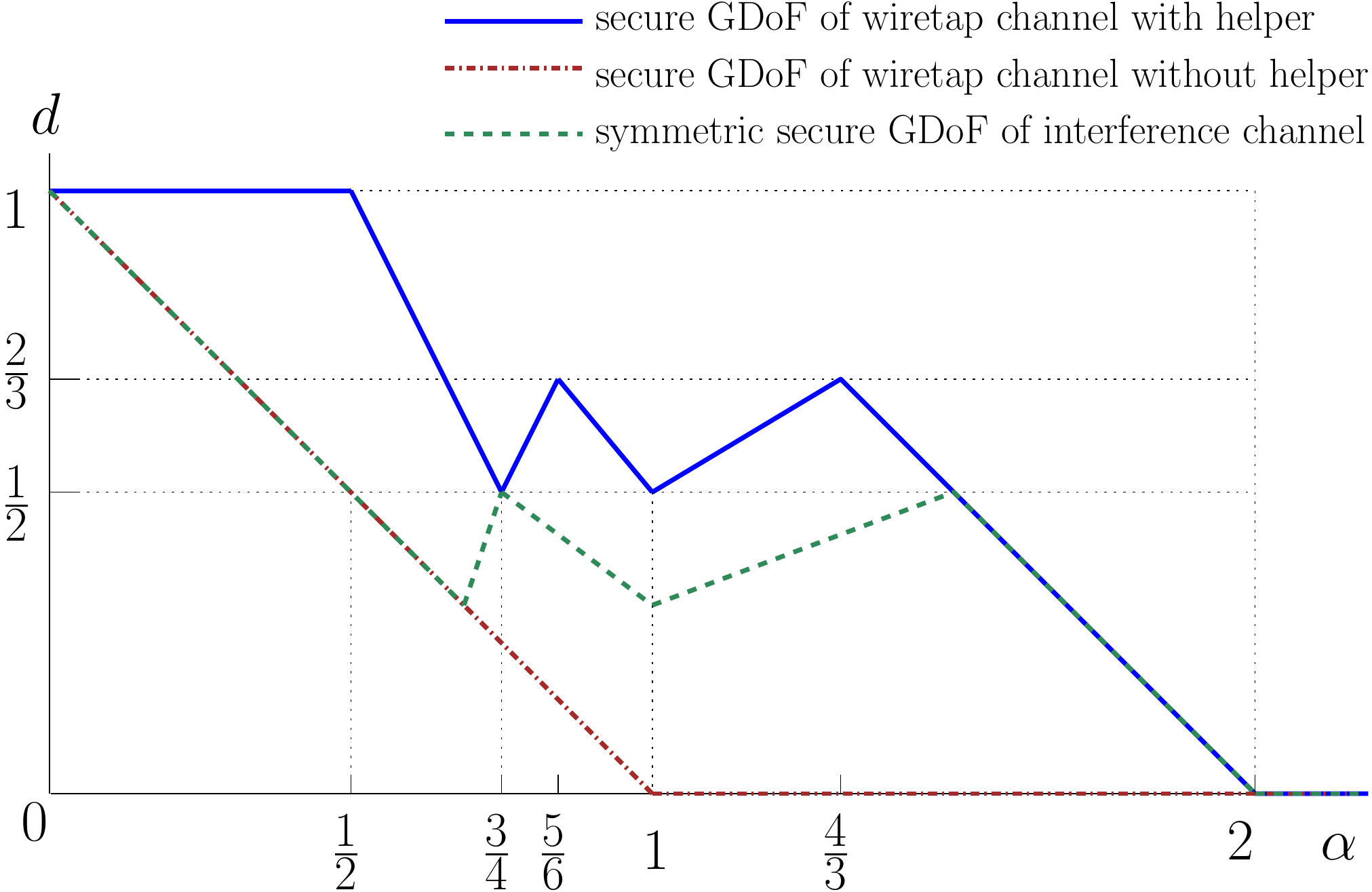}
\caption{Optimal secure  GDoF vs. $\alpha$ for the symmetric  Gaussian  wiretap channels with and without a helper, where $\alpha$ denotes the ratio between the direct-links'  channel strength  (in decibel scale) and the cross-links'  channel strength  (in decibel scale).   The direct-links refers to the link from the  transmitter to its legitimate receiver, as well as the link from the helper to the eavesdropper, while the   cross-links refers to the other two links, in this symmetric setting. The dash line with green color refers to  the optimal symmetric secure GDoF vs. $\alpha$ for the symmetric  Gaussian  interference  channel with confidential messages (cf.~\cite{ChenAllerton:18}).}
\label{fig:SGDoFWTH}
\end{figure} 

Specifically, the main  contribution of this work is the \emph{optimal} secure GDoF  characterization of a symmetric wiretap channel with a helper, for all the channel parameters.  The result reveals that, adding a helper can significantly increase the secure GDoF of the wiretap channel (see Fig.~\ref{fig:SGDoFWTH}).  The result is supported by a new converse and a new scheme.  The  converse is derived for the  wiretap channel with a helper under the \emph{general} channel parameters, i.e., the converse holds for the symmetric and asymmetric channels.  
In the proposed scheme,\footnote{Although in this paper, for illustration simplicity, the achievable scheme is mainly discussed for the symmetric settings, the key ideas could be generalized to asymmetric channels as well.}
  the helper sends a cooperative jamming signal at a specific power level and direction. In this way, it minimizes the penalty in GDoF incurred by the secrecy constraint.  In the proposed scheme, the  signal of the transmitter is a superposition of  a common signal, a middle signal, and a private signal. The power of  private signal is low enough such that this signal arrives at the eavesdropper under the noise level.  The power of the common signal and middle signal is above the noise level at the receivers. However, at the eavesdropper each of the common signal and middle signal is aligned  at a specific power level and direction with the jamming signal sent from the helper, which minimizes the penalty in GDoF incurred by the secrecy constraint.
The optimal secure GDoF is described in  different expressions for different interference regimes. For each interference regime, the power  and rate levels of the signals in the proposed scheme are set to the optimal values, so as to achieve the optimal secure GDoF. 
In the secure rate analysis, the techniques of noise removal and signal separation  are used (cf.~\cite{MGMK:14, NM:13}).  
The secure GDoF result derived in this work can be extended to understand the secure  capacity to within a constant gap, which will be investigated in the future work.

We will organize the rest of this work as follows.  
In Section~\ref{sec:system} we will describe the channel model.   In Section~\ref{sec:mainresult}, the  main results of this work will be provided. 
The converse proof will be described  in   Section~\ref{sec:converse}.
The achievability proof will be shown  in   Section~\ref{sec:CJGau} and Section~\ref{sec:epcase4}.
In Section~\ref{sec:concl333} we will provide the conclusion.
In terms of notations, we  use $\Hen(\bullet)$ and $\Imu(\bullet)$  to represent the  entropy and mutual information,  respectively, and use $\hen(\bullet)$  to represent differential entropy.  
$\Zc$ and  $\Zc^+$ are used to   denote the sets of integers and  positive integers, respectively, while   $\Rc$ is used to  denote the set of real numbers.   
$(\bullet)^+= \max\{0, \bullet\}$. When $f(s)=o(g(s))$ is used, it suggests that $\lim_{s \to \infty} f(s)/g(s) =0$. 
All the logarithms are considered with base~$2$.

\section{System model  \label{sec:system} }

This work focuses on a Gaussian wiretap channel with a helper (see Fig.~\ref{fig:WTH}). In this setting,   transmitter~1 sends a \emph{confidential} message  to receiver~1 (the legitimate  receiver), in the presence of a helper (transmitter~2) and an eavesdropper (receiver~2).  
By following the common conventions in \cite{NM:13,ChenAllerton:18, ChenIC:18},  the channel input-output relationship of this setting is described by
\begin{subequations} \label{eq:Cmodel} 
\begin{align}
y_{1}(t) &= \sqrt{P^{\alpha_{11}}} h_{11} x_{1}(t) +   \sqrt{P^{\alpha_{12}}} h_{12} x_{2}(t) +z_{1}(t)   \\ 
y_{2}(t) &= \sqrt{P^{\alpha_{21}}} h_{21} x_{1}(t) +   \sqrt{P^{\alpha_{22}}} h_{22} x_{2}(t) +z_{2}(t)   
\end{align}
\end{subequations}
where $y_{k}(t)$  denotes the received signal of receiver~$k$ at time $t$,    $x_{k}(t)$  denotes the transmitted signal of  transmitter~$k$  with a normalized power constraint  $\E |x_{k}(t)|^2 \leq 1$,  and $z_k(t) \sim \mathcal{N}(0, 1)$ denotes the additive white Gaussian noise. 
$ \sqrt{P^{\alpha_{k\ell}}} h_{k\ell}$ represents the channel gain of the  link from  transmitter~$\ell$ to receiver~$k$, for $ \ell, k =1,2$. 
The nonnegative  parameter  $\alpha_{k\ell}$ captures the \emph{link strength} of the  channel from  transmitter~$\ell$ to receiver~$k$.  
$h_{k\ell} \in (1, 2]$  is a parameter of the channel gain. 
 In this setting, $P \geq 1$ captures the base of signal-to-noise ratio of all the links.  
Since  the form of $ \sqrt{P^{\alpha_{k\ell}}} h_{k\ell}$ can describe any real channel gain greater than $1$,   the above model can describe the  general channels (focusing on the cases with channel gains greater than $1$) in terms of capacity approximation.
In this setting,  all the nodes  are assumed to know all the channel parameters  $\{\alpha_{k\ell}, h_{k\ell}\}_{k, \ell}$.\footnote{The availability of the CSI of the eavesdropper links can be justified when the eavesdropper is a legitimate user in the network, as in the case of interference channels with confidential messages \cite{XU:14,ChenIC:18}.}  
When we consider the \emph{symmetric} case,  we will assume that 
\[  \alpha_{12}= \alpha_{21}=\alpha,  \quad \alpha_{22}= \alpha_{11}=1,  \quad   \alpha \geq 0.  \]

\begin{figure}[t!]
\centering
\includegraphics[width=6cm]{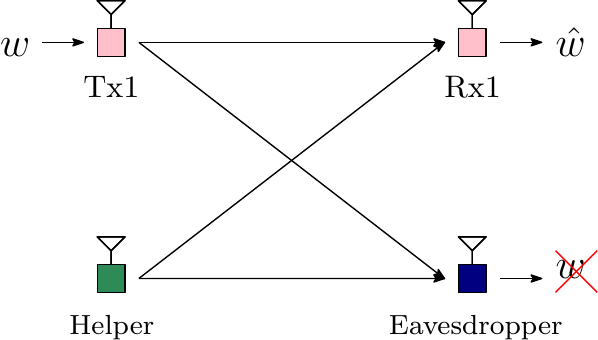}
 \vspace{-.05 in}
\caption{A Gaussian wiretap channel with a helper, where  transmitter~1 sends a \emph{confidential} message  $w$ to receiver~1, in the presence of a helper (transmitter~2) and an eavesdropper (receiver~2).}
\label{fig:WTH}
\end{figure}

In  this  setting, transmitter~$1$  sends a message $w$ to its legitimate  receiver over $\bln$ channel uses,  where $w$  is   chosen uniformly from the set $\mathcal{W} \defeq \{1,2, 3, \cdots, 2^{\bln R}\}$. 
When transmitting the confidential message from transmitter~1,   a stochastic function  \[f_1: \mathcal{W}_0 \times \mathcal{W}   \to   \mathcal{R} ^{\bln} \] maps $w \in \mathcal{W}$  to the signal  $ x_1^{\bln}  = f_1(w_0, w)   \in  \mathcal{R} ^{\bln}$, where the randomness in this mapping is represented by $w_0 \in  \mathcal{W}_0$. We assume that $w_0$ and $w$ are independent. 
At the helper (transmitter~$2$), another function \[f_2: \mathcal{W}_h  \to   \mathcal{R}^{\bln} \]  generates a random signal $ x_2^{\bln}  = f_2(w_h)   \in  \mathcal{R} ^{\bln}$, where    $w_h \in  \mathcal{W}_h$ is a random variable that is independent of $w_0$ and $w$.  We assume that $w_0$ is available  at the first transmitter only, while $w_h$ is  available  at the second transmitter only.
We say a secure rate $R$ is achievable if there exists a sequence of  codes with $\bln$-length,  such that the legitimate  receiver  can reliably decode the message $w$, i.e.,  
 \begin{align}
 \text{Pr}[w  \neq \hat{w} ]  \leq \epsilon   \label{eq:Pedef}
  \end{align}
  for any $\epsilon >0$, 
and the message is secure from the eavesdropper, i.e.,  
 \begin{align}
 \Imu(w; y_{2}^{\bln})  &\leq  \bln \epsilon .   \label{eq:secrecyconstraint} 
 \end{align}
 The above secrecy  constraint is also known as a weak notion of secrecy constraint. 
We will use $C(\alpha_{11}, \alpha_{12}, \alpha_{21}, \alpha_{22}, P)$ to denote the  secure capacity, which is defined as the maximal  secure rate that is achievable.
We will use  $d (\alpha_{11}, \alpha_{12}, \alpha_{21}, \alpha_{22})$ to denote the  secure  GDoF, which is defined as 
 \begin{align}
 d (\alpha_{11}, \alpha_{12}, \alpha_{21}, \alpha_{22})\defeq   \lim_{P \to \infty}   \frac{C(\alpha_{11}, \alpha_{12}, \alpha_{21}, \alpha_{22}, P)}{\frac{1}{2}\log P}.  \label{eq:GDoF99}
 \end{align}
GDoF is a form of the approximation of capacity.  In this setting, DoF is a particular case of GDoF by considering  only one point with $\alpha_{12}= \alpha_{21}= \alpha_{22}= \alpha_{11}=1$.

\section{Main result  \label{sec:mainresult}}

The main result of this paper is the characterization of the optimal secure GDoF value for the symmetric wiretap channel with a helper defined in Section~\ref{sec:system}.

\begin{theorem}  \label{thm:GDoF}
Consider  the \emph{symmetric}  Gaussian  wiretap channel with a helper  defined in Section~\ref{sec:system}, with   $\alpha_{12}= \alpha_{21}=\alpha$ and $\alpha_{22}= \alpha_{11}=1$. For almost all channel coefficients  $\{h_{k\ell}\} \in (1, 2]^{2\times 2}$,  the optimal secure  GDoF  is characterized as 
\begin{subnumcases}
{d(\alpha) =} 
  \quad  1   &                              for   \ $ 0 \leq \alpha \leq  1/2$        \label{thm:capacitydet1} \\
 \quad 2 - 2\alpha  &  for \ $1/2  \leq \alpha \leq 3/4$    \label{thm:capacitydet2} \\ 
  \quad 2\alpha -1  &  for \ $3/4  \leq \alpha \leq  5/6$    \label{thm:capacitydet3} \\ 
 \quad      3/2 - \alpha  &                            for  \          $5/6 \leq \alpha \leq 1$    \label{thm:capacitydet4} \\ 
   \quad       \alpha/2  &  for  \  $1  \leq \alpha \leq 4/3$    \label{thm:capacitydet6} \\ 
     \quad                   2- \alpha  &  for \   $4/3 \leq  \alpha \leq 2$   \label{thm:capacitydet7} \\ 
     \quad                                          0  &  for  \   $2\leq  \alpha$ .  \label{thm:capacitydet8}
\end{subnumcases}
\end{theorem}
\begin{proof}
The converse follows from Lemma~\ref{lm:gupper} and Corollary~\ref{cor:symSGDoF} described in Section~\ref{sec:converse}. 
Specifically, Lemma~\ref{lm:gupper} provides some  upper bounds on the secure \emph{rate} of the  Gaussian  wiretap channel with a helper, under \emph{general} channel parameters.  Corollary~\ref{cor:symSGDoF} is the GDoF result derived from Lemma~\ref{lm:gupper},  in the setting of \emph{symmetric} Gaussian  wiretap channel with a helper.
The achievability based on cooperative jamming is described in Sections~\ref{sec:CJGau} and~\ref{sec:epcase4}.
\end{proof}

\vspace{10pt}

\begin{remark}
In Fig.~\ref{fig:SGDoFWTH}, for the wiretap channel \emph{without} a helper (removing transmitter~2), the secure GDoF, denoted by $d_{no} $, is given by
\begin{align}
d_{no}  (\alpha) = (1 -  \alpha)^+       \quad \forall  \alpha \in [0,  \infty)   \non 
\end{align}
(see Appendix~\ref{sec:GDoFnoHelper} for details). Comparing $d_{no}  (\alpha)$ with $d(\alpha)$ in Theorem \ref{thm:GDoF}, one can find that adding a helper significantly increases the secure GDoF of the wiretap channel.
\end{remark}

\begin{remark}
From Fig.~\ref{fig:SGDoFWTH}, it is not hard to verify that for all $\alpha$ values, the sum secure GDoF of two-user symmetric interference channels with confidential messages is no less than the secure GDoF of symmetric wiretap channel with a helper (surprisingly, for a large regime of $\frac{3}{2}\leq\alpha<2$, the former is the double of the latter). This indicates that in a two-user symmetric interference channel with confidential messages, acting one of the transmitter as a pure helper (i.e., not sending its own confidential message) does not  improve the sum secure GDoF. 
\end{remark}


Our achievability scheme is based on   pulse  amplitude modulation, superposition coding, cooperative jamming, and  alignment techniques.  Specifically, the transmitted signal of transmitter~1 is a superposition of  a \emph{common signal} (denoted by $v_c$), a \emph{middle signal} (denoted by $v_m$), and a \emph{private signal} (denoted by $v_p$). 
The power of  private signal is low enough such that this signal arrives at the eavesdropper under the noise level.  The power of the common signal is higher than that of the middle signal, and both signals are above the noise level at the receivers.
The transmitted signal of the helper (transmitter~2) is a superposition of  a \emph{common jamming signal} (denoted by $u_c$) and a \emph{private jamming signal} (denoted by $u_p$).
The common signal $v_c$  is aligned with the common jamming signal  $u_c$ at the  eavesdropper at a specific power level and direction, while the middle signal $v_m$  is aligned with the private jamming signal  $u_p$ at the  eavesdropper, again,  at a specific power level and direction.  
One difference between the common and  private jamming signals is that, the latter  arrives at the legitimate receiver under the noise level. 
Therefore, the signal $u_p$  will \emph{not} cause much  interference at the legitimate receiver,  while the signal $u_c$ will  create significant interference at the legitimate receiver.
For the \emph{private signal}, it will not require jamming signal from the helper. 
Depending on different regimes of $\alpha$, some signals are not needed and are thus set as  empty signals. Table~\ref{tab:signals} gives a summary of the signal design. Note that  $(v_c, u_c)$ is considered as a pair of signals, and $(v_m, u_p)$ is considered as another pair of signals, due to our alignment design. 
With the above signal design, our scheme achieves the optimal secure GDoF value. More details of the achievability scheme can be found in  Sections~\ref{sec:CJGau} and~\ref{sec:epcase4}.

\begin{table}
\caption{Signal design for different cases, where ``$\checkmark$'' means that the signal will be sent and  ``$\times$''  means that the signal won't be sent.}
\begin{center}
{\renewcommand{\arraystretch}{1.7}
\begin{tabular}{|c|c|c|c|c|c|c|}
  \hline
                     &   $0 \leq \alpha \leq \frac{1}{2}$  &  $\frac{1}{2}  \leq  \alpha \leq \frac{3}{4}$  &   $\frac{3}{4}  \leq \alpha \leq \frac{5}{6}$   & $\frac{5}{6} \leq \alpha \leq 1 $  &  $1 \leq  \alpha \leq \frac{4}{3}$  & $\frac{4}{3} \leq \alpha \leq 2$     \\
    \hline
   $(v_c, u_c) $    		&   $\times$     			&   $\times$                       &     $\checkmark$       &    $\checkmark$     &   $\checkmark$    &    $\checkmark$   \\
    \hline
   $(v_m, u_p) $    		&   $\checkmark$     			&    $\checkmark$                        &     $\checkmark$       &   $\checkmark$     &   $\times$     &    $\times$    \\
    \hline
   $v_p$    		&   $\checkmark$   			&   $\checkmark$                       &    $\checkmark$      &   $\checkmark$   &   $\times$     &    $\times$     \\
    \hline
       \end{tabular}
}
\end{center}
\label{tab:signals}
\end{table}

\section{Converse \label{sec:converse}}

For the  Gaussian  wiretap channel with a helper  defined in Section~\ref{sec:system},  we provide a  general upper bound on the secure  rate, which is stated in the following  Lemma~\ref{lm:gupper}.

\begin{lemma}  \label{lm:gupper}
For the  Gaussian  wiretap channel with a helper  defined in Section~\ref{sec:system},  letting $\phi_1 \defeq  ( \alpha_{12}- (\alpha_{22} - \alpha_{21})^+)^+ $ and $\phi_3 \defeq \min\{\alpha_{21},  \alpha_{12},  ( \alpha_{11} - \phi_1)^+ \}$,  the secure  rate is upper bounded   by 
\begin{align}
R  & \leq    \frac{1}{2} \log \Bigl(1+ P^{\alpha_{11} - \phi_3} \cdot \frac{|h_{11}|^2}{|h_{21}|^2} +  P^{\alpha_{12} -(\alpha_{2 2} -\alpha_{2 1}  )^+} \cdot \frac{|h_{12}|^2}{|h_{22}|^2}  \Bigr)      +   \frac{1}{2} \log \bigl(1+      P^{ \phi_3  - \phi_1  }  |h_{22}|^2   \bigr)   +   7.3   \label{eq:upWT}   \\ 
  R  & \leq  \frac{1}{2}  \Bigl(   \frac{1}{2} \log \bigl(1+   \frac{P^{ (\alpha_{11} - \alpha_{21} )^+}}{|h_{21}|^2}   \bigr)   +   \frac{1}{2} \log \bigl(1+   \frac{P^{ (\alpha_{22} - \alpha_{12} )^+}}{|h_{12}|^2}   \bigr)  + \frac{1}{2} \log \bigl(1+ P^{\alpha_{11}}   |h_{11}|^2  + P^{\alpha_{12}}   |h_{12}|^2  \bigr)       +  \log 9   \Bigr) \label{eq:detup3}  \\
    R  & \leq      \frac{1}{2} \log \Bigl(1+ P^{\alpha_{11}-\alpha_{21}} \cdot \frac{|h_{11}|^2}{|h_{21}|^2} +  P^{\alpha_{22} + \alpha_{11} - \alpha_{21} } \cdot \frac{|h_{11}|^2 |h_{22}|^2}{|h_{21}|^2}  \Bigr)       \label{eq:detup1} .
\end{align}
\end{lemma}

 \vspace{5pt}

The proof of Lemma~\ref{lm:gupper} is provided in the following subsections. 
Our converse  is based on genie-aided techniques. Specifically, we enhance the setting by using the following two approaches. a)  Adding noise with a certain power to the observation of eavesdropper. See    \eqref{eq:ups62366} later on, where $y_{2}(t)$ becomes $\bar{y}_{2}(t)$, which is a noisy version of $y_{2}(t)$.  b) Adding information to the legitimate receiver. See \eqref{eq:ups22det} and \eqref{eq:2R1add223} later on. 
The converse also builds on  some bounds on the difference of conditional differential
entropy of an interference channel (see Lemmas~\ref{lm:differencefirst2} and \ref{lm:differencea11} later on). More details of the converse proof are provided in the following subsections.

Based on Lemma~\ref{lm:gupper}, we provide the  secure GDoF upper bound in the following corollary.
\begin{corollary}  \label{cor:SGDoF}
For the Gaussian   wiretap channel with a helper  defined in Section~\ref{sec:system},  the secure  GDoF is upper bounded   by 
\begin{align}
d (\alpha_{11}, \alpha_{12}, \alpha_{21}, \alpha_{22}) \leq   \min\Bigl\{  & \underbrace{ \max \{  \phi_1    ,  (\alpha_{11} -  \phi_3)^+ \} +  (\phi_3   -    \phi_1 )^+}_{ Bound~1} ,  \non\\ &  \underbrace{\frac{1}{2}\bigl( (\alpha_{11} - \alpha_{21} )^+ +  (\alpha_{22} - \alpha_{12} )^+  +  \max\{ \alpha_{11}, \alpha_{12}  \}\bigr)}_{ Bound~2} ,    \non\\&   \underbrace{  (\alpha_{22} + \alpha_{11} - \alpha_{21})^+  }_{ Bound~3}    \Bigr\}  .
\end{align}
\end{corollary}
\begin{proof}
The first bound $d   \leq   \max \{  \phi_1    ,  (\alpha_{11} -  \phi_3)^+ \} +  (\phi_3   -    \phi_1 )^+$ follows from the bound in \eqref{eq:upWT}. The second bound follows from the bound in \eqref{eq:detup3} and the last bound follows from the  bound in \eqref{eq:detup1}.
\end{proof}

The following result is a simplified version of Corollary~\ref{cor:SGDoF} for the symmetric setting with $\alpha_{11}= \alpha_{22}=1,   \alpha_{21}= \alpha_{12}=\alpha$. 
\begin{corollary}  \label{cor:symSGDoF}
For the symmetric Gaussian  wiretap channel with a helper  defined in Section~\ref{sec:system}, with $\alpha_{11}= \alpha_{22}=1,   \alpha_{21}= \alpha_{12}=\alpha$, the secure  GDoF is upper bounded  by 
\begin{subnumcases}
{d  (\alpha) \leq } 
  \quad  1   &                              for   \ $ 0 \leq \alpha \leq  1/2$        \non \\
 \quad 2 - 2\alpha  &  for \ $1/2  \leq \alpha \leq 3/4$    \non \\ 
  \quad 2\alpha -1  &  for \ $3/4  \leq \alpha \leq  5/6$    \non \\ 
 \quad      3/2 - \alpha  &                            for  \          $5/6 \leq \alpha \leq 1$   \non \\ 
   \quad       \alpha/2  &  for  \  $1  \leq \alpha \leq 4/3$   \non \\ 
     \quad                   2- \alpha  &  for \   $4/3 \leq  \alpha \leq 2$   \non \\ 
     \quad                                          0  &  for  \   $2\leq  \alpha$ .  \non 
\end{subnumcases}
\end{corollary}
\begin{proof}
See Appendix~\ref{sec:symSGDoF}. 
\end{proof}
\vspace{10pt}

In what follows, we provide the proof of Lemma~\ref{lm:gupper}. 
At first,   for $k,\ell   \in \{1,2\}, k \neq \ell$ we define that 
\begin{align}
    \phi_1 &\defeq  (\alpha_{12}- (\alpha_{22 } - \alpha_{21})^+ )^+   \label{eq:defphi1}  \\
     \phi_2 &\defeq  (\alpha_{11} - \phi_1)^+    \label{eq:defphi2} \\
     \phi_3 &\defeq \min\{\alpha_{21},  \alpha_{12},  \phi_2 \}  \label{eq:defphi3} \\
    s_{kk}(t)  & \defeq   \sqrt{P^{(\alpha_{kk}-\alpha_{k\ell})^+}} h_{kk} x_{k}(t) +  \tilde{z}_{k}(t)    \label{eq:defs11} \\
 s_{\ell k}(t)  & \defeq   \sqrt{P^{\alpha_{\ell k}}} h_{\ell k} x_{k}(t)  +z_{\ell}(t)  \label{eq:defs12}  \\
  \bar{x}_{1}(t)  &\defeq  \sqrt{ P^{    \min\{\alpha_{21},  \alpha_{12},  \alpha_{11} -   \phi_1\}   }}  h_{21} x_{1}(t)  +   \bar{z}_{3}(t)  \label{eq:defx1b} \\
            \bar{x}_{2}(t) & \defeq     \sqrt{P^{  \phi_3}}   \tilde{z}_{2}(t)   +   \bar{z}_{4}(t)     \label{eq:defx2b}   
 \end{align}
 and 
 \begin{align}
           \bar{y}_{2}(t)  &\defeq  \sqrt{P^{-(\alpha_{2 1} - \phi_3 )}} y_{2}(t)    +   \bar{z}_{2}(t)   \non\\
      &=  \sqrt{P^{\phi_3 }}  h_{21} x_{1}(t) +  \sqrt{P^{\alpha_{2 2} -(\alpha_{2 1} - \phi_3 )}}   h_{22} x_{2}(t)    +  \sqrt{P^{-(\alpha_{2 1} - \phi_3 )}}  z_{2}(t)
         +   \bar{z}_{2}(t)    \label{eq:defybdet} 
\end{align}
 where  $\tilde{z}_{1}(t), \tilde{z}_{2}(t), \bar{z}_{2}(t), \bar{z}_{3}(t), \bar{z}_{4}(t) \sim \mathcal{N}(0, 1)$  are  i.i.d.  noise random variables  that are independent of the other noise random variables and transmitted signals $\{x_1(t), x_2 (t)\}_t$.  Let $s^{\bln}_{kk} \defeq \{ s_{kk}(t)\}_{t=1}^{\bln}$, $s^{\bln}_{\ell k} \defeq \{ s_{\ell k}(t)\}_{t=1}^{\bln}$,  $\bar{x}^{\bln}_{k} \defeq \{ \bar{x}_{k}(t)\}_{t=1}^{\bln}$ and $\bar{y}^{\bln}_{2} \defeq \{ \bar{y}_{2}(t)\}_{t=1}^{\bln}$.

 \subsection{Proof of bound \eqref{eq:upWT} \label{sec:upWT} }

Let us focus on the proof of bound \eqref{eq:upWT}.  For the channel defined in Section~\ref{sec:system}, the secure rate is bounded as follows:
\begin{align}
  \bln R &= \Hen(w)    \nonumber \\
  &= \Imu(w; y^{\bln}_1) + \Hen(w| y^{\bln}_1)  \nonumber \\
  &\leq  \Imu(w; y^{\bln}_1) + \bln \epsilon_{1,n}    \label{eq:Fanodet}  \\
  &\leq  \Imu(w;  y^{\bln}_1) - \Imu(w; y^{\bln}_2) +  \bln \epsilon_{1,n} + \bln \epsilon   \label{eq:secrecydet}   \\ 
  &\leq  \Imu(w;  y^{\bln}_1,  s_{22}^{\bln}) - \Imu(w; y^{\bln}_2) +  \bln \epsilon_{1,n} + \bln \epsilon   \label{eq:ups22det}  \\
    &=    \underbrace{\Imu(w;  s_{22}^{\bln})}_{=0} +  \Imu(w;  y^{\bln}_1|  s_{22}^{\bln})  - \Imu(w; y^{\bln}_2) +  \bln \epsilon_{1,n} + \bln \epsilon   \nonumber  \\
    &=     \Imu(w;  y^{\bln}_1|  s_{22}^{\bln})  - \Imu(w; y^{\bln}_2) +  \bln \epsilon_{1,n} + \bln \epsilon   \label{eq:ups2244det}  \\
     &\leq      \Imu(w;  y^{\bln}_1|  s_{22}^{\bln})  - \Imu(w;   \bar{y}_{2}^{\bln}) +  \bln \epsilon_{1,n} + \bln \epsilon   \label{eq:ups62366}   
\end{align}
where \eqref{eq:Fanodet} follows from  Fano's inequality; 
\eqref{eq:secrecydet} results from secrecy constraint in \eqref{eq:secrecyconstraint}, i.e.,  $ \Imu(w; y^{\bln}_2)  \leq \bln \epsilon$ for an  arbitrary small  $\epsilon$;
 \eqref{eq:ups22det}  uses the fact that adding information will not reduce the mutual information;
\eqref{eq:ups2244det} follows  from the fact  that $w$ is independent of  $ x_{2}^{\bln}$  and $ s_{22}^{\bln}$ (cf.~\eqref{eq:defs11});
\eqref{eq:ups62366} stems from the fact that  $w \to y^{\bln}_2  \to  \bar{y}_{2}^{\bln}$ forms a Markov chain, which implies that $\Imu(w;   y^{\bln}_2) \geq  \Imu(w;   \bar{y}_{2}^{\bln})$.

We invoke the following lemma to bound  $\Imu(w;   \bar{y}_{2}^{\bln})$ appeared in \eqref{eq:ups62366}.

  \vspace{10pt}
\begin{lemma}  \label{lm:bound1112} 
For  $s_{22}(t)$ and $\bar{y}_{2}(t)$ defined  in  \eqref{eq:defs11} and \eqref{eq:defybdet},  the following inequality holds true
\begin{align}
\Imu(w;   \bar{y}_{2}^{\bln})  \geq   \Imu( w;  \bar{y}_{2}^{\bln} | s_{22}^{\bln})     -   \frac{\bln}{2}  \log  14  . \label{eq:bound1112}    
\end{align}
\end{lemma}
\begin{proof}
The proof of this lemma is provided in Appendix~\ref{sec:bound1112}.
\end{proof}

Then, by incorporating the result of Lemma~\ref{lm:bound1112} into \eqref{eq:ups62366}, it gives 
 \begin{align}
 & \bln R  - \frac{\bln}{2}  \log 14  -  \bln \epsilon_{1,n} - \bln \epsilon     \non\\
  \leq &    \Imu(w;  y^{\bln}_1|  s_{22}^{\bln})   - \Imu( w;  \bar{y}_{2}^{\bln} | s_{22}^{\bln})  \non   \\
     =  &    \hen( y^{\bln}_1|  s_{22}^{\bln})  - \hen( \bar{y}_{2}^{\bln} | s_{22}^{\bln})  + \hen(   \bar{y}_{2}^{\bln} | s_{22}^{\bln}, w) -  \hen(  y^{\bln}_1|  s_{22}^{\bln}, w)      \non    \\
         = &     \hen( y^{\bln}_1|  s_{22}^{\bln})  - \hen( \bar{y}_{2}^{\bln} | s_{22}^{\bln})  + \hen(     \bar{y}_{2}^{\bln} , s_{22}^{\bln} |  w)     -   \hen(  y^{\bln}_1, s_{22}^{\bln}| w)    \label{eq:ups9266}   
\end{align}
where  \eqref{eq:ups9266} uses the identities that $\hen(   \bar{y}_{2}^{\bln} | s_{22}^{\bln}, w) =   \hen(     \bar{y}_{2}^{\bln} , s_{22}^{\bln} |  w)   -  \hen( s_{22}^{\bln} |  w) $ and $  \hen(  y^{\bln}_1|  s_{22}^{\bln}, w)  = \hen(  y^{\bln}_1, s_{22}^{\bln}| w) - \hen( s_{22}^{\bln}| w)$.
For the first two terms  in \eqref{eq:ups9266},  we have an upper bound that is stated in the following lemma. 
 
 \vspace{10pt}
\begin{lemma}  \label{lm:differencefirst2} 
For  $s_{22}(t)$ and $\bar{y}_{2}(t)$ defined  in  \eqref{eq:defs11} and \eqref{eq:defybdet}, the following inequality holds true
 \begin{align}
    \hen( y^{\bln}_1|  s_{22}^{\bln})  - \hen( \bar{y}_{2}^{\bln} | s_{22}^{\bln})    \leq &   \frac{n}{2} \log \Bigl(1+ P^{\alpha_{11} - \phi_3} \cdot \frac{|h_{11}|^2}{|h_{21}|^2} +  P^{\alpha_{12} -(\alpha_{2 2} -\alpha_{2 1}  )^+} \cdot \frac{|h_{12}|^2}{|h_{22}|^2}  \Bigr)     +    \frac{n}{2} \log 10  . \non
\end{align}
\end{lemma}
\begin{proof}
The proof of this lemma is provided in Appendix~\ref{sec:differencefirst2}.
\end{proof}

 For the last two terms in \eqref{eq:ups9266},  we have an upper bound that is stated in the following lemma. 
 
 \vspace{10pt}
\begin{lemma}  \label{lm:differencea11} 
For  $y_{2}(t)$ defined  in \eqref{eq:Cmodel}, and $s_{22}(t) $  defined in \eqref{eq:defs11},  the following inequality holds true
\begin{align}
\hen(     \bar{y}_{2}^{\bln} , s_{22}^{\bln} |  w)     -   \hen(  y^{\bln}_1, s_{22}^{\bln}| w) 
\leq   \frac{n}{2} \log \bigl(1+      P^{ \phi_3  - \phi_1  }  |h_{22}|^2 \bigr)      +  \frac{n}{2} \log 168    . \label{eq:differencea11}    
\end{align}
\end{lemma}
\begin{proof}
The proof of this lemma is provided in Appendix~\ref{sec:differencea11}.
\end{proof}

Finally, by incorporating  the results of Lemmas~\ref{lm:differencefirst2}  and \ref{lm:differencea11}  into \eqref{eq:ups9266},  the secure rate is bounded by 
 \begin{align}
     R  
     &\leq    \frac{1}{2} \log \Bigl(1+ P^{\alpha_{11} - \phi_3} \cdot \frac{|h_{11}|^2}{|h_{21}|^2} +  P^{\alpha_{12} -(\alpha_{2 2} -\alpha_{2 1}  )^+} \cdot \frac{|h_{12}|^2}{|h_{22}|^2}  \Bigr)       +   \frac{1}{2} \log \bigl(1+      P^{ \phi_3  - \phi_1  }  |h_{22}|^2   \bigr)   +   7.3   +    \epsilon_{1,n}  + \epsilon     .   \non  
\end{align} 
Letting  $n\to \infty, \  \epsilon_{1,n} \to 0, \ \epsilon_{2,n} \to 0$ and  $\epsilon \to 0$, we  get the desired bound \eqref{eq:upWT}.

\subsection{Proof of bound \eqref{eq:detup3} \label{sec:detup3} }

Let us now prove the bound \eqref{eq:detup3}.    
Let \[\tilde{x}_{k}(t) \defeq \sqrt{P^{\max\{\alpha_{k k}, \alpha_{\ell k}\}}}x_{k}(t) +  \tilde{z}_{k}(t) \]  and  $ \tilde{x}^{\bln}_{k} \defeq \{ \tilde{x}_{k}(t)  \}_{t=1}^{\bln} $ for $k,\ell   \in \{1,2\}, k \neq \ell$,  where  $\tilde{z}_{k}(t) \sim \mathcal{N}(0, 1)$ is a virtual noise that is independent of the other noise and transmitted signals. Recall that \[s_{\ell k}(t)  \defeq   \sqrt{P^{\alpha_{\ell k}}} h_{\ell k} x_{k}(t)  +z_{\ell}(t)\] for $k,\ell   \in \{1,2\}, k \neq \ell$ (cf.~\eqref{eq:defs12})
Beginning with Fano's inequality, the secure rate  is bounded as:
\begin{align}
  \bln R -  \bln \epsilon_{1,n}   
&\leq  \Imu(w; y^{\bln}_1)   \non \\
  &\leq  \Imu(w; y^{\bln}_1)    - \Imu(w; y^{\bln}_2) + \bln \epsilon    \label{eq:2R1sec325}  \\
  &\leq  \Imu(w; y^{\bln}_1, \tilde{x}^{\bln}_{1}, \tilde{x}^{\bln}_{2},    y^{\bln}_2 )   - \Imu(w; y^{\bln}_2) + \bln \epsilon    \label{eq:2R1add223}  \\
        &=   \hen(\tilde{x}^{\bln}_{1}, \tilde{x}^{\bln}_{2} ) -  \hen(y^{\bln}_2 )  + \hen(y^{\bln}_1,y^{\bln}_2 | \tilde{x}^{\bln}_{1}, \tilde{x}^{\bln}_{2} )  - \hen(y^{\bln}_1, \tilde{x}^{\bln}_{1}, \tilde{x}^{\bln}_{2} |   y^{\bln}_2, w )   + \bln \epsilon   \non  \\
        &\leq    \hen(\tilde{x}^{\bln}_{1}, \tilde{x}^{\bln}_{2} ) -  \hen(s^{\bln}_{21} )  + \hen(y^{\bln}_1,y^{\bln}_2 | \tilde{x}^{\bln}_{1}, \tilde{x}^{\bln}_{2} )  - \hen(y^{\bln}_1, \tilde{x}^{\bln}_{1}, \tilde{x}^{\bln}_{2} |   y^{\bln}_2, w )   + \bln \epsilon     \label{eq:2R1Markove8836}   
 \end{align}
where \eqref{eq:2R1sec325}  results  from a secrecy constraint (cf.~\eqref{eq:secrecyconstraint});
\eqref{eq:2R1add223}  stems from the fact that adding information does not decrease the mutual information;
  \eqref{eq:2R1Markove8836}  follows from the derivation that $ \hen(y^{\bln}_2 ) \geq  \hen(y^{\bln}_2 |x^{\bln}_2) = \hen(s^{\bln}_{21} )$.  Note that $y_2(t) =    \sqrt{P^{\alpha_{22}}} h_{22} x_{2}(t) +   s_{21}(t) $. 
 On the other hand, we have 
\begin{align}
   \bln R   &\leq  \Imu(x^{\bln}_{1}; y^{\bln}_1 )  + \bln \epsilon_{1,n}   \label{eq:2R1Markove82435}  \\
  &=    \hen(y^{\bln}_1 )  -  \hen( s^{\bln}_{12} | x^{\bln}_{1})   +  \bln \epsilon_{1,n}     \label{eq:2R1y122}  \\
  &=    \hen(y^{\bln}_1 )  -  \hen( s^{\bln}_{12})   + \bln \epsilon_{1,n}     \label{eq:2R1ind254}  
 \end{align}
where  
\eqref{eq:2R1Markove82435}  follows from  the Markov chain  of $w \to  x^{\bln}_{1}  \to  y^{\bln}_1$;  
\eqref{eq:2R1y122} results from the fact that $y_1(t) =    \sqrt{P^{\alpha_{11}}} h_{11} x_{1}(t) +   s_{12}(t) $;
\eqref{eq:2R1ind254}  follows from the independence between $x^{\bln}_{1}$ and $s^{\bln}_{12}$.
Finally, by combining \eqref{eq:2R1Markove8836} and \eqref{eq:2R1ind254}, it gives
\begin{align}
 &  2\bln R  -  2\bln \epsilon_{1,n} - \bln \epsilon    \non  \\
 &\leq    \hen(\tilde{x}^{\bln}_{1}) -  \hen( s^{\bln}_{21})  +   \hen( \tilde{x}^{\bln}_{2} ) -  \hen( s^{\bln}_{12})  + \hen(y^{\bln}_1 )     + \hen(y^{\bln}_1,y^{\bln}_2 | \tilde{x}^{\bln}_{1}, \tilde{x}^{\bln}_{2} )  - \hen(y^{\bln}_1, \tilde{x}^{\bln}_{1}, \tilde{x}^{\bln}_{2} |   y^{\bln}_2, w ).    \label{eq:2R1sum8256}
 \end{align}
 At this point, by following the steps from (171)-(176) in \cite{ChenIC:18}, 
 we end up with
\begin{align}
 &  2 R +   2 \epsilon_{1,n}  -  \epsilon    \non  \\
      &\leq   \frac{1}{2} \log \bigl(1+   \frac{P^{ (\alpha_{11} - \alpha_{21} )^+}}{|h_{21}|^2}   \bigr)    +   \frac{1}{2} \log \bigl(1+   \frac{P^{ (\alpha_{22} - \alpha_{12} )^+}}{|h_{12}|^2}   \bigr)  + \frac{1}{2} \log \bigl(1+ P^{\alpha_{11}}   |h_{11}|^2  + P^{\alpha_{12}}   |h_{12}|^2  \bigr)        +  \log 9 .  \non
 \end{align}
By setting  $n\to \infty, \  \epsilon_{1,n} \to 0 $ and  $\epsilon \to 0$, it gives bound \eqref{eq:detup3}.

\subsection{Proof of bound \eqref{eq:detup1} \label{sec:detup12} }

Bound \eqref{eq:detup1} is directly from \cite[Lemma~8]{ChenIC:18}.

\section{Achievability  \label{sec:CJGau} }

This section focuses on  the \emph{symmetric} Gaussian wiretap channel with a helper defined in Section~\ref{sec:system}. 
Note that we consider the symmetric channel mainly for illustration simplicity. The key ideas of the achievable scheme presented in this section could be generalized to asymmetric settings as well. 
For this channel, we will provide a cooperative jamming scheme to  achieve the optimal secure GDoF expressed in Theorem~\ref{thm:GDoF}.  
The proposed scheme will use   pulse  amplitude modulation (PAM) and signal alignment.  
The details of the scheme are described as follows. 

 \subsubsection{Codebook}
At transmitter~$1$, a codebook is  generated as
  \begin{align}
     \mathcal{B} \defeq \Bigl\{  v^{\bln} (w,  w_0):  \  w \in \{1,2,\cdots, 2^{\bln R}\},   
      w_0 \in \{1,2,\cdots, 2^{\bln R_0}\}   \Bigr\}     \label{eq:code2341J}
     \end{align}
with $v^{\bln} $ being the codewords. All the codewords' elements are independent and identically generated according to a specific distribution. 
  $R$ and  $R_0$ are the rates  of  the confidential message $w$  and the confusion message $w_0$, respectively.  The purpose of using the confusion message is to  guarantee the security of the  confidential message.
  The message will be mapped to a codeword under the following two steps.  First, given the message $w$, a  sub-codebook  $\mathcal{B}( w) $ is selected, where   $\mathcal{B}( w) $  is   defined  as 
\[   \mathcal{B} (w)  \defeq \bigl\{ v^{\bln} (w,  w_0): \  w_0 \in \{1,2,\cdots, 2^{\bln R_0}\}   \bigr\}.   \]
Second,   transmitter~1 \emph{randomly} selects a codeword from  the selected sub-codebook   based on  a uniform distribution.
Then, the channel input is mapped from selected codeword $v^{\bln}$  such that 
 \begin{align}
  x_1 (t) =   h_{22}  v (t)      \label{eq:xvkkk}      
   \end{align}
  for  $t=1,2, \cdots, \bln$, where $v (t)$ denotes  the $t$th element of $v^{\bln}$.

\subsubsection{Constellation and  alignment}
In the proposed scheme, each codeword  $v^{\bln} $ is generated  such that each element takes the following form
 \begin{align}
   v (t)  =     \sqrt{P^{ - \beta_{c}}}  \cdot  v_{c} (t) + \sqrt{P^{ - \beta_{m}}}  \cdot  v_{m} (t) +   \sqrt{P^{ - \beta_{p}}}    \cdot  v_{p}(t)   \label{eq:xvk}  
 \end{align}
 which suggests that   the  input $ x_1$  in~\eqref{eq:xvkkk} can be described as 
 \begin{align}
  x_1  = &    \sqrt{P^{ - \beta_{c}}}   h_{22} v_{c}  +  \sqrt{P^{ - \beta_{m}}}   h_{22}v_{m}   +   \sqrt{P^{ - \beta_{p} }}    h_{22} v_{p} \label{eq:xvkkk1}  
 \end{align}
without the time index for simplicity (same for the next signal descriptions).  For transmitter~2 (the helper), the transmitted signal is  a cooperative jamming signal designed as  
    \begin{align}
  x_2  = &    \sqrt{P^{ \alpha - 1 - \beta_{c} }}   h_{21}  u_{c} +    \sqrt{P^{ \alpha - 1 - \beta_{m}  }}   h_{21}  u_{p} .     \label{eq:u3def}  
 \end{align}
For the above transmitted signals,  the random variables $v_{c}$, $v_{m}$, $v_{p}$, $u_{c}$ and $u_{p}$ are  \emph{independently} (cross symbols and times) and \emph{uniformly}  drawn from the corresponding PAM constellation sets
 \begin{align}
   v_{c} ,  u_{c}     &  \in    \Omega ( \xi =  \frac{  6 \gamma}{Q} ,   \   Q =  P^{ \frac{ \lambda_{c} }{2}} )  \label{eq:constellationGsym1}   \\ 
     v_{m} ,  u_{p}     &  \in    \Omega ( \xi =   \frac{ 2 \gamma}{Q} ,   \   Q =  P^{ \frac{ \lambda_{m} }{2}} )  \label{eq:constellationGsym3}   \\ 
      v_{p}      &  \in    \Omega ( \xi =  \frac{ \gamma}{Q} ,   \   Q = P^{ \frac{  \lambda_{p} }{2}} )    \label{eq:constellationGsym2}    
 \end{align}
 where   $\Omega (\xi,  Q)  \defeq   \{ \xi \cdot a :   \    a \in  \Zc  \cap [-Q,   Q]   \}$  denotes the PAM constellation set,   and  $\gamma$ is a  finite constant such that
  \begin{align}
\gamma  \in (0, 1/ 20  ].   \label{eq:gammadef} 
 \end{align}
In   Table~\ref{tab:wthpara} we provide the  parameters $\{\beta_{c}, \beta_{m},\beta_{p},  \lambda_{c},  \lambda_{m} ,   \lambda_{p} \}$  under different regimes\footnote{Without loss of generality, we assume that $P^{ \frac{ \lambda_{c} }{2}}$, $P^{ \frac{ \lambda_{m} }{2}}$ and $P^{ \frac{  \lambda_{p} }{2}}$ are integers. Consider one example with  $\lambda_{c} = \alpha  -1/2 - \epsilon$ and $P^{ \frac{ \lambda_{c} }{2}}=\sqrt{ P^{ \alpha  -1/2 - \epsilon}}$. When $\sqrt{ P^{ \alpha  -1/2 - \epsilon}}$ is not an integer, we can slightly modify $\epsilon$ such that $\sqrt{ P^{ \alpha  -1/2 - \epsilon}}$ is an integer, for the regime with large $P$.}.  
If the parameters are set as $\beta_{p} = \infty$ and $\lambda_{p}=0$, we will treat  $ v_{p}$ as an empty term  in the transmitted signal. Similar implication is applied to $\{   v_{c} ,  u_{c} ,   v_{m} ,  u_{p} \}$ .
In Fig.~\ref{fig:schematic_WTH} we provide a schematic representation of the proposed scheme, focusing on  transmitter~1.

 \begin{figure}[h]
\centering
\includegraphics[width=16cm]{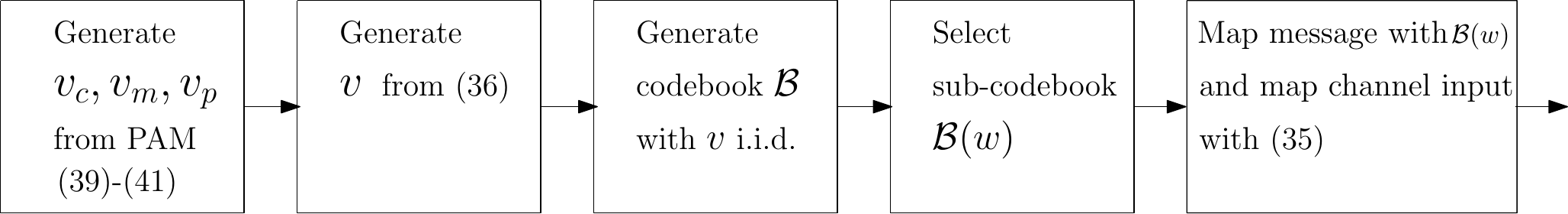}
\caption{A schematic representation of the proposed scheme, focusing on transmitter~1.}
\label{fig:schematic_WTH}
\end{figure}

Given our signal design,   the power constraints  $\E |x_1|^2 \leq 1$ and $\E |x_2|^2 \leq 1$ are satisfied.  
Focusing on the first transmitter, we have  
\begin{align}
 \E |v_{c}|^2   &=  \frac{2 \times (\frac{  6 \gamma}{Q})^2 }{ 2Q +1}  \sum_{i=1}^{Q} i^2  =  \frac{  (\frac{  6 \gamma}{Q})^2 \cdot Q(Q+1)}{3}  \leq   \frac{  72 \gamma^2 }{3}   \non\\   
  \E |v_{m}|^2    &\leq   \frac{  8 \gamma^2 }{3}  \non\\ 
  \E |v_{p}|^2    &\leq   \frac{  2  \gamma^2 }{3}  
\end{align} 
which suggests that 
\begin{align}
 \E |x_1|^2 \leq  4\times  (\frac{  72 \gamma^2 }{3} + \frac{  8 \gamma^2 }{3}   +  \frac{  2  \gamma^2 }{3}    )  = \frac{328}{3}  \gamma^2 \leq  \frac{328}{3}  \times  \frac{1}{ 400}  <1 .
\end{align} 
Similarly, we have $\E |x_2|^2 \leq 1$. 
Note that, with our parameter design it holds true that $\beta_{c} \geq  \alpha - 1$ and   $\beta_{m} \geq 2 \alpha - 1$, which controls the average power of the transmitted signal $x_2$  to satisfy $\E |x_2|^2 \leq 1$.

The above signal design then implies the following forms of the received signals
\begin{align}
y_{1} &=    \sqrt{P^{ 1 - \beta_{c}}} h_{11}    h_{22} v_{c} +   \sqrt{P^{ 1 - \beta_{m} }} h_{11}    h_{22} v_{m}  +    \sqrt{P^{ 1 - \beta_{p} }} h_{11}    h_{22}  v_{p}  \non\\& \quad  +     \sqrt{P^{ 2 \alpha - 1 -  \beta_{c}  }}  h_{12}  h_{21}  u_c  +     \sqrt{P^{ 2 \alpha - 1 -  \beta_{m}  }}  h_{12}  h_{21}  u_p   +  z_{1}    \label{eq:yvk1}  \\
y_{2} &=       h_{21}    h_{22}  (   \sqrt{P^{ \alpha - \beta_{c}}} (  v_{c}    +     u_c   ) +  \sqrt{P^{ \alpha - \beta_{m}}} (   v_{m}    +     u_p ) )+    \sqrt{P^{ \alpha - \beta_{p}}}  h_{21}    h_{22}  v_{p}        +  z_{2} .   \label{eq:yvk2}  
\end{align}
As we can see,  at the eavesdropper, the  jamming signal $u_c$  (resp. $u_p$) is aligned at  a specific power level and direction  with the  signal $v_{c}$  (resp. $v_{m}$).  In this way, it will minimize the penalty in GDoF incurred by the secrecy constraint, which can  be seen later. 
Note that, with the above parameter design, the power of signal term  with $v_{p}$ is under the noise level at receiver~2, while the power of signal term  with $u_{p}$ is under the noise level at receiver~1.

\subsubsection{Secure rate analysis} For  $\epsilon >0$, let us define  the two rates   as
\begin{align}
R &\defeq   \Imu(v; y_1) -  \Imu ( v; y_{2} ) - \epsilon   \label{eq:Rk623J} \\  
R_0  &\defeq  \Imu ( v; y_{2}) - \epsilon .   \label{eq:Rk623bJ}  
\end{align}
Note that the wiretap channel with a helper can be considered as  a specific case of the two-user interference channel with confidential messages, by removing the message of the second transmitter. Therefore,  the result of \cite[Theorem~2]{XU:15} reveals that the rate $R$ defined in \eqref{eq:Rk623J} is achievable  and  the message $w$ is secure, i.e.,   $\Imu(w; y_{2}^{\bln})  \leq  \bln \epsilon$.

In the following, we will analyze the secure rate for different cases of $\alpha$.

\begin{table}
\caption{Designed  parameters for different cases, for some $\epsilon >0$.}
\begin{center}
{\renewcommand{\arraystretch}{1.7}
\begin{tabular}{|c|c|c|c|c|c|c|}
  \hline
                     &   $0 \leq \alpha \leq \frac{1}{2}$  &  $\frac{1}{2}  \leq  \alpha \leq \frac{3}{4}$  &   $\frac{3}{4}  \leq \alpha \leq \frac{5}{6}$   & $\frac{5}{6} \leq \alpha \leq 1 $  &  $1 \leq  \alpha \leq \frac{4}{3}$  & $\frac{4}{3} \leq \alpha \leq 2$     \\
   \hline
   $\beta_{c} $    		&   $\infty$     			&   $ \infty $                     &    $0 $    &   $0 $ &   $ \alpha - 1$    &    $\alpha - 1 $   \\
    \hline
       $\beta_{m} $    		&   $0$     			&   $  2\alpha -1  $    &    $2\alpha -1 $    &   $2\alpha -1  $ &   $ \infty$    &    $\infty$   \\
    \hline
       $\beta_{p} $    		&   $\alpha$     			&   $  \alpha$    &       $\alpha $    &   $\alpha $ &   $ \infty$    &    $\infty $   \\
    \hline 
   $\lambda_{c}$ 		&   $0$ 	&  $0 $              &   $4\alpha -3    - \epsilon$     &  $  \alpha  -1/2    - \epsilon$  &  $\alpha/2  - \epsilon$     &   $2- \alpha - \epsilon$        \\
  \hline
     $\lambda_{m}$ 		&   $\alpha  - \epsilon$ 	&  $1 - \alpha    - \epsilon $     &   $1- \alpha   - \epsilon$     &  $  1  -  \alpha  - \epsilon$  &  $0$     &   $0$        \\
  \hline
   $\lambda_{p}$  &   $1 - \alpha - \epsilon$  	&  $1- \alpha - \epsilon$     &     $1- \alpha - \epsilon $  &   $ 1- \alpha - \epsilon$ &  $0$     &     $0 $    \\
    \hline
    \end{tabular}
}
\end{center}
\label{tab:wthpara}
\end{table}

 \subsection{Rate analysis when $0\leq  \alpha \leq  1/2$   \label{sec:CJscheme012}}

 For the first case with $0\leq  \alpha \leq  1/2$,   the parameter design in Table~\ref{tab:wthpara} gives the following forms of the transmitted signals
  \begin{align}
  x_1  = &      h_{22} v_{m}   +   \sqrt{P^{ - \alpha }}    h_{22} v_{p}  \label{eq:xvkkk129552}    \\
  x_2  = &    \sqrt{P^{ \alpha - 1 }}   h_{21}  u_{p} .        \label{eq:u3def252525}  
 \end{align}
 Then, the received signals become
 \begin{align}
y_{1} &=    \sqrt{P} h_{11}    h_{22} v_{m}   +    \sqrt{P^{ 1 - \alpha }} h_{11}    h_{22}  v_{p}    +     \sqrt{P^{ 2 \alpha - 1   }}  h_{12}  h_{21}  u_p    +  z_{1}    \label{eq:yvk12866ww}  \\
y_{2} &=      \sqrt{P^{ \alpha }}  h_{21}    h_{22}   (  v_{m}    +     u_p   )+     h_{21}    h_{22}  v_{p}        +  z_{2} .   \label{eq:yvk27r9r99r}  
\end{align}
Let us now analyze the achievable secure rate expressed in \eqref{eq:Rk623J}, i.e.,    
\begin{align}
R  & =     \Imu(v; y_1) -  \Imu ( v; y_{2} )      \label{eq:lbrate1} 
\end{align}
by setting  $\epsilon \to 0$. 
We will begin with the first term in the right-hand side of \eqref{eq:lbrate1}. 
With our signal design,  $v$  is now expressed as $v = v_{m} +   \sqrt{P^{ - \alpha }}  v_{p}$.
In this case,   the two random variables $v_{m}$ and $v_{p}$ can be estimated  from  $y_1$, with  error probability denoted by $\text{Pr} [  \{ v_{m} \neq \hat{v}_{m} \} \cup  \{ v_{p} \neq \hat{v}_{p} \}  ] $.  
For the  first term in the right-hand side of \eqref{eq:lbrate1}, we have the following bound
  \begin{align}
  \Imu(v; y_1)   &\geq   \Imu(v; \hat{v}_{m}, \hat{v}_{p})  \label{eq:rate5544}     \\
  &=   \Hen(v) -   \Hen(v  |  \hat{v}_{m}, \hat{v}_{p})    \non    \\
    &\geq    \Hen(v) -       \bigl( 1+    \text{Pr} [  \{ v_{m} \neq \hat{v}_{m} \} \cup  \{ v_{p} \neq \hat{v}_{p} \}  ] \cdot \Hen(v) \bigr)  \label{eq:rate28536}     \\
        & =     \bigl( 1 -   \text{Pr} [  \{ v_{m} \neq \hat{v}_{m} \} \cup  \{ v_{p} \neq \hat{v}_{p} \}  ] \bigr)   \cdot \Hen(v)  - 1  \label{eq:rate2256}     
 \end{align}
where  \eqref{eq:rate5544} uses  the Markov property of $v \to y_1 \to  \{\hat{v}_{m}, \hat{v}_{p} \}$; and 
\eqref{eq:rate28536} follows from Fano's inequality.
 The entropy  $\Hen(v)$ in \eqref{eq:rate2256} can be computed as 
   \begin{align}
 \Hen(v) & =  \Hen(v_{m}) +   \Hen(v_{p})  \non\\ 
        &=  \log (2 \cdot P^{ \frac{ \alpha  - \epsilon}{2}} +1)    +  \log (2 \cdot P^{ \frac{ 1 - \alpha - \epsilon}{2}} +1)   \non\\
        &= \frac{ 1  - 2\epsilon}{2} \log P + o(\log P)        \label{eq:rate39737}                              
 \end{align}  
 using the facts that    $v_{m} \in    \Omega (\xi   =  2 \gamma \cdot \frac{ 1}{Q},   \   Q =  P^{ \frac{ \alpha  - \epsilon}{2}} ) $ and $v_{p}  \in    \Omega (\xi   =\gamma \cdot \frac{ 1}{Q},   \   Q = P^{ \frac{ 1 - \alpha - \epsilon}{2}} ) $,  and  that   $\{ v_{p}, v_{m}\}$ can be reconstructed from $v$,  and vice versa. 
For the error probability  appeared in \eqref{eq:rate2256}, we have the following result.
 \begin{lemma}  \label{lm:errorcase1}
Consider the case with $0\leq  \alpha \leq  1/2$, and consider the signal design in  \eqref{eq:xvkkk1}-\eqref{eq:gammadef} and Table~\ref{tab:wthpara}. Then, the error probability of the estimation  of  $\{v_{m}, v_{p} \}$ from $y_1$ is
 \begin{align}
\text{Pr} [  \{ v_{m} \neq \hat{v}_{m} \} \cup  \{ v_{p} \neq \hat{v}_{p} \}  ]    \to 0         \quad \text {as}\quad  P\to \infty .   \label{eq:epcase1}
 \end{align}
 \end{lemma}
\begin{proof}
The proof is described in Appendix~\ref{sec:errorcase1}. In the proof, a successive decoding method is used in the estimation  of  $\{v_{m}, v_{p} \}$ from $y_1$.
\end{proof}

By combining the results of \eqref{eq:rate2256}, \eqref{eq:rate39737}  and \eqref{eq:epcase1},  it produces the following bound 
  \begin{align}
\Imu(v; y_1)  &\geq   \frac{ 1  - 2\epsilon}{2} \log P + o(\log P) .   \label{eq:rb286227}  
 \end{align}
 Note that the  term $ \Imu ( v; y_{2} )$ in  the right-hand side of \eqref{eq:lbrate1} can be considered as a penalty term in the secure rate, incurred by the secrecy constraint. We will show that, with our scheme design using signal alignment, 
 this penalty will be minimized to a small value that can be ignored in terms of GDoF. It can be seen from the received signal of the eavesdropper  that, the jamming signal is aligned at a specific power level and direction with the  signal sent from transmitter~1 (see \eqref{eq:yvk2}). 
Let us now bound  the penalty  term $ \Imu ( v; y_{2} )$ as follows: 
  \begin{align}
  \Imu ( v; y_{2} )    
\leq &  \Imu(v; y_2,  v_{m} + u_p  )    \non \\
=   & \Imu(v;    v_{m} + u_p  )   +  \Imu(v;   h_{21}    h_{22}  v_{p}        +  z_{2}   |  v_{m} + u_p )    \non   \\
=   & \Hen( v_{m} + u_p  )  - \Hen(u_p  )      +  \hen( h_{21}    h_{22}  v_{p}        +  z_{2}  )   -  \hen(   z_{2} )   \non    \\
\leq &   \underbrace{\log (4  \cdot P^{ \frac{ \alpha - \epsilon}{2}} +1)    -   \log (2  \cdot P^{ \frac{ \alpha  - \epsilon}{2}} +1) }_{\leq 1 }    +   \frac{1}{2} \log ( 2 \pi e (  \underbrace{| h_{21}|^2  | h_{22}|^2      +  1 }_{\leq 17}))       -  \frac{1}{2}\log (2 \pi e)    \label{eq:rb2526}   \\
\leq  & \log (2\sqrt{17})    \label{eq:rb26277}  
 \end{align}
where  
\eqref{eq:rb2526} results from the identity  that  Gaussian input maximizes the differential entropy and the fact that $v_{m} + u_p \in  \Omega (\xi   =  2 \gamma \cdot P^{ - \frac{ \alpha  - \epsilon}{2}},   \   Q = 2 P^{ \frac{ \alpha  - \epsilon}{2}} ) $. Note that uniform distribution maximizes the entropy. 
At the final step,  we  incorporate the results of  \eqref{eq:rb286227}  and \eqref{eq:rb26277} into \eqref{eq:lbrate1} and then  get the following bound on the secure rate  
\begin{align}
R         \geq   \frac{ 1  - 2\epsilon}{2} \log P + o(\log P).    \label{eq:rbcase1f}   
\end{align}
It implies that the secure GDoF  $d =  1$ is achievable  for this case  with   $0 \leq   \alpha \leq  1/2$.

\subsection{Rate analysis when $1/2\leq  \alpha \leq  3/4$   \label{sec:CJscheme1234}}
 
 Given the parameter design in Table~\ref{tab:wthpara}, in this case the transmitted signals are simplified as 
  \begin{align}
  x_1  = &     \sqrt{P^{ -(2\alpha -1)}}   h_{22}v_{m}   +   \sqrt{P^{ - \alpha }}    h_{22} v_{p} \label{eq:xvkkk139747}   \\
  x_2  = &        \sqrt{P^{ - \alpha  }}   h_{21}  u_{p}      \label{eq:u3def858594}  
 \end{align}
 which gives  the following forms of the received signals 
\begin{align}
y_{1} &=       \sqrt{P^{ 2- 2\alpha }} h_{11}    h_{22} v_{m}  +    \sqrt{P^{ 1 - \alpha }} h_{11}    h_{22}  v_{p}  +    h_{12}  h_{21}  u_p   +  z_{1}    \label{eq:yvk13637}  \\
y_{2} &=     \sqrt{P^{ 1- \alpha  }}  h_{21}    h_{22}   (   v_{m}    +     u_p ) +    h_{21}    h_{22}  v_{p}        +  z_{2} .   \label{eq:yvk2373737}  
\end{align}
In this case, we can prove that the secure rate  $R  \geq   \frac{ 2 - 2 \alpha -  2\epsilon}{2} \log P + o(\log P)  $ is achievable.  The rate analysis for this case follows from the steps in the previous case (cf.~\eqref{eq:lbrate1}-\eqref{eq:rbcase1f}). 
To avoid the repetition, we will just provide the outline of the proof for this case.  In the first step, it can be proved that 
 \begin{align}
\Imu(v; y_1)  &\geq   \frac{ 2 - 2 \alpha -  2\epsilon}{2} \log P + o(\log P)    \label{eq:rb26262}  
 \end{align}
 by following the derivations in \eqref{eq:lbrate1}-\eqref{eq:rb286227}. In this case $v =  \sqrt{P^{ -(2 \alpha -1) }}  v_{m} +   \sqrt{P^{ - \alpha }}  v_{p}$ and  $\Hen(v) = \Hen(v_{m}) +   \Hen(v_{p})   =  \frac{ 2 - 2 \alpha -  2\epsilon}{2} \log P + o(\log P)$.   Similarly to the conclusion in  Lemma~\ref{lm:errorcase1} for the previous case,  in this case it is also true  that  the error probability of the estimation  of  $\{v_{m}, v_{p} \}$ from $y_1$  vanishes as $P\to \infty$. A successive decoding method is also used in this estimation.
 In the second step, by following the derivations related to \eqref{eq:rb2526} and \eqref{eq:rb26277},    it can be proved that   
     \begin{align}
 \Imu ( v; y_{2} )   &\leq    o(\log P),    \label{eq:rb29537}  
 \end{align} 
 which, together with \eqref{eq:rb26262}, gives the lower bound on the secure rate $R  \geq   \frac{ 2 - 2 \alpha -  2\epsilon}{2} \log P + o(\log P)  $.  It implies that the secure GDoF  $d = 2- 2\alpha $ is achievable  for this case.

  \subsection{Rate analysis when $3/4\leq  \alpha \leq  5/6$   \label{sec:CJscheme3456}}
 
  Given the parameter design in Table~\ref{tab:wthpara}, in this case the transmitted signals are simplified as 
  \begin{align}
  x_1  = &      h_{22} v_{c}  +  \sqrt{P^{ - (2 \alpha - 1)}}   h_{22}v_{m}   +    \sqrt{P^{ - \alpha }}     h_{22} v_{p} \label{eq:xvkkk1285267}   \\
  x_2  = &    \sqrt{P^{ \alpha - 1  }}   h_{21}  u_{c} +    \sqrt{P^{ - \alpha   }}   h_{21}  u_{p}     .  \label{eq:u3def3i638}  
 \end{align}
 The received signals are given by 
 \begin{align}
y_{1} &=    \sqrt{P} h_{11}    h_{22} v_{c} +     \sqrt{P^{ 2 \alpha - 1  }}  h_{12}  h_{21}  u_c  +   \sqrt{P^{ 2 -2 \alpha }} h_{11}    h_{22} v_{m}  +    \sqrt{P^{ 1 - \alpha }} h_{11}    h_{22}  v_{p}     +     h_{12}  h_{21}  u_p   +  z_{1}    \label{eq:yvk19637374}  \\
y_{2} &=       h_{21}    h_{22}  (   \sqrt{P^{ \alpha }} (  v_{c}    +     u_c   ) +  \sqrt{P^{ 1- \alpha }} (   v_{m}    +     u_p ) )+    h_{21}    h_{22}  v_{p}        +  z_{2} .   \label{eq:yvk2846949}  
\end{align}
The rate analysis also follows~\eqref{eq:lbrate1}-\eqref{eq:rbcase1f}.
In this case, we can prove that the secure rate  $R  \geq   \frac{  2 \alpha -1 -  3\epsilon}{2} \log P + o(\log P)  $ is achievable.  
Again, to avoid the repetition we will just provide the outline of the proof. 
In the first step, it can be proved that 
 \begin{align}
\Imu(v; y_1)  &\geq   \frac{ 2 \alpha -1 -  3\epsilon}{2} \log P + o(\log P)    \label{eq:rb23526}  
 \end{align}
 by following the derivations in \eqref{eq:lbrate1}-\eqref{eq:rb286227}. 
Here we have $v =  v_{c} + \sqrt{P^{ -(2 \alpha -1) }} v_{m} +   \sqrt{P^{ - \alpha }}  v_{p}$ and  $\Hen(v) =\Hen(v_{c}) +   \Hen(v_{m}) +   \Hen(v_{p})   =  \frac{2 \alpha -1 -  3\epsilon}{2} \log P + o(\log P)$.  It is  true  that  the error probability of the estimation  of  $\{v_{c}, v_{m}, v_{p} \}$ from $y_1$  vanishes as $P\to \infty$. A successive decoding method is also used in this estimation.
 In the second step, by following the derivations related to \eqref{eq:rb2526} and \eqref{eq:rb26277},   we can prove that
     \begin{align}
 \Imu ( v; y_{2} )   &\leq    o(\log P) .   \label{eq:rb29537ee}  
 \end{align} 
  Therefore, the secure rate is bounded by $R \geq   \frac{2 \alpha -1 -  3\epsilon}{2} \log P + o(\log P)  $, implying that the secure GDoF  $d = 2 \alpha -1 $ is achievable.

  \subsection{Rate analysis when $5/6\leq  \alpha \leq  1$   \label{sec:CJscheme561}}
 
 In this case, the transmitted signals take the same forms as in \eqref{eq:xvkkk1285267} and \eqref{eq:u3def3i638}, and the received signals are expressed as in \eqref{eq:yvk19637374} and  \eqref{eq:yvk2846949}.  However,  in the rate analysis, the estimation approach is different, where noise removal and signal separation will be used.
 By following the previous derivations  in \eqref{eq:rate5544}-\eqref{eq:rate2256}, we have the following bound 
 \begin{align}
\Imu(v; y_1) 
  \geq  \bigl( 1 -   \text{Pr} [ \{ v_{c} \neq \hat{v}_{c} \} \cup  \{ v_{m} \neq \hat{v}_{m} \} \cup  \{ v_{p} \neq \hat{v}_{p} \}  ] \bigr)   \cdot \Hen(v)  - 1 \label{eq:rate925626}     
 \end{align} 
 where the entropy $\Hen(v)$ in \eqref{eq:rate925626} can be computed as 
   \begin{align}
 \Hen(v)  =  \Hen(v_{c}) + \Hen(v_{m})  +   \Hen(v_{p})  
        = \frac{ 3/2  - \alpha - 3\epsilon}{2} \log P + o(\log P) .       \label{eq:rate39737ee}                              
 \end{align}  
The following lemma shows that the error probability $\text{Pr} [ \{ v_{c} \neq \hat{v}_{c} \} \cup  \{ v_{m} \neq \hat{v}_{m} \} \cup  \{ v_{p} \neq \hat{v}_{p} \}  ] $ in \eqref{eq:rate925626} vanishes as $P$ approaches infinity.
 \begin{lemma}  \label{lm:rateerror561}
Consider the case with $5/6\leq  \alpha \leq  1$, and consider the signal design in  \eqref{eq:xvkkk1}-\eqref{eq:gammadef} and Table~\ref{tab:wthpara}. Then, for almost all the channel realizations  $\{h_{k\ell}\} \in (1, 2]^{2\times 2}$, the error probability of the estimation  of  $\{ v_{c}, v_{m}, v_{p} \}$ from $y_1$ is
 \begin{align}
\text{Pr} [ \{ v_{c} \neq \hat{v}_{c} \} \cup  \{ v_{m} \neq \hat{v}_{m} \} \cup  \{ v_{p} \neq \hat{v}_{p} \}  ]    \to 0         \quad \text {as}\quad  P\to \infty .   \label{eq:epcase4}
 \end{align}
 \end{lemma}
\begin{proof}
The proof is described in Section~\ref{sec:epcase4}. In the proof,  noise removal and signal separation are used in the estimation  of  $\{ v_{c}, v_{m}, v_{p} \}$ from $y_1$.
\end{proof}
 The results of \eqref{eq:rate925626}, \eqref{eq:rate39737ee}  and \eqref{eq:epcase4} imply that the bound 
  \begin{align}
\Imu(v; y_1)  &\geq   \frac{ 3/2  - \alpha - 3\epsilon}{2} \log P + o(\log P)    \label{eq:rb26278}  
 \end{align}
 holds true  for almost all the channel realizations  $\{h_{k\ell}\} \in (1, 2]^{2\times 2}$. 
Next, by following the derivations related to \eqref{eq:rb2526} and \eqref{eq:rb26277},    it can be proved that   
     \begin{align}
 \Imu ( v; y_{2} )   &\leq    o(\log P) .   \label{eq:rb34677}  
 \end{align} 
 Thus, we have $R  \geq  \frac{ 3/2  - \alpha - 3\epsilon}{2} \log P + o(\log P)$, implying that the secure GDoF  $d = 3/2  - \alpha  $ is achievable, for almost all the channel realizations in this case.

\subsection{Rate analysis when $1\leq  \alpha \leq  4/3$   \label{sec:CJscheme143}}
 
In this case, the transmitted signals are simplified as 
 \begin{align}
  x_1  = &    \sqrt{P^{ - (\alpha - 1)  }}    h_{22} v_{c}       \label{eq:xvkkk102627}   \\
  x_2  = &      h_{21}  u_{c}       \label{eq:u3def373894}  
 \end{align}
 and the received signals are expressed as
 \begin{align}
y_{1} &=    \sqrt{P^{ 2 - \alpha  }}   h_{11}    h_{22} v_{c}    +     \sqrt{P^{  \alpha  }}  h_{12}  h_{21}  u_c    +  z_{1}    \label{eq:yvk12867}  \\
y_{2} &=       \sqrt{P}   h_{21}    h_{22}  (  v_{c}    +     u_c   )         +  z_{2} .   \label{eq:yvk283882}  
\end{align}
As in the previous case,  the estimation approaches of noise removal and signal separation are also used here for the rate analysis.   
 In this case,  the entropy $\Hen(v)$ is computed as  $\Hen(v)  =  \Hen(v_{c})  = \frac{    \alpha/2 - \epsilon}{2} \log P + o(\log P) $.
 Following the previous derivations  in \eqref{eq:rate5544}-\eqref{eq:rate2256}, we have
 \begin{align}
\Imu(v; y_1) &  \geq  \bigl( 1 -   \text{Pr} [ \{ v_{c} \neq \hat{v}_{c} \}   ] \bigr)   \cdot \Hen(v)  - 1 \non\\
  &=    \bigl( 1 -   \text{Pr} [ \{ v_{c} \neq \hat{v}_{c} \} ] \bigr)   \cdot  \frac{    \alpha/2 - \epsilon}{2} \log P   + o(\log P) .
   \label{eq:rate34647}     
 \end{align} 
 The  lemma below provides the result on the  error  probability $\text{Pr} [ v_{c} \neq \hat{v}_{c}  ] $. 
 \begin{lemma}  \label{lm:epcase5}
Consider the case with $1\leq  \alpha \leq  4/3$, and consider the signal design in  \eqref{eq:xvkkk1}-\eqref{eq:gammadef} and Table~\ref{tab:wthpara}. Then,  for almost all the channel realizations  $\{h_{k\ell}\} \in (1, 2]^{2\times 2}$,  the error probability of the estimation  of  $ v_{c}$ from $y_1$ is
 \begin{align}
 \text{Pr} [  v_{c} \neq \hat{v}_{c}  ]   \to 0         \quad \text {as}\quad  P\to \infty .   \label{eq:epcase5}
 \end{align}
 \end{lemma}
\begin{proof}
The proof is described in Appendix~\ref{sec:epcase5}. In the proof,  noise removal and signal separation are used in the estimation  of  $ v_{c}$ from $y_1$.
\end{proof}
 The results of \eqref{eq:rate34647} and \eqref{eq:epcase5} implies that  the following bound 
  \begin{align}
\Imu(v; y_1)  &\geq \frac{    \alpha/2 - \epsilon}{2} \log P + o(\log P)    \label{eq:rb36727}  
 \end{align}
 holds true  for almost all the channel realizations  $\{h_{k\ell}\} \in (1, 2]^{2\times 2}$. 
 Again, it is not hard to prove that $\Imu ( v; y_{2} )  \leq    o(\log P)$. Together with \eqref{eq:rb36727}, it reveals that  $R  \geq  \frac{    \alpha/2 - \epsilon}{2} \log P + o(\log P)$, and that the secure GDoF  $d =  \alpha/2 $ is achievable,  for almost all the channel realizations in this case.

 \subsection{Rate analysis when $4/3 \leq  \alpha \leq  2$   \label{sec:CJscheme432}}
 
 In this case, the transmitted signals take the same forms as in \eqref{eq:xvkkk102627} and \eqref{eq:u3def373894}, and the received signals are expressed as in \eqref{eq:yvk12867} and  \eqref{eq:yvk283882}. 
Here  the entropy $\Hen(v)$ is computed as  $\Hen(v)  =  \Hen(v_{c})  = \frac{   2- \alpha - \epsilon}{2} \log P + o(\log P) $.
Similar to the previous cases, we can prove that 
 \begin{align}
\Imu(v; y_1) &  \geq     \bigl( 1 -   \text{Pr} [ v_{c} \neq \hat{v}_{c} ] \bigr)   \cdot  \frac{   2- \alpha - \epsilon}{2} \log P + o(\log P)
   \label{eq:rate34647ee}     
 \end{align} 
 (cf.~\eqref{eq:rate34647}).
The probability $\text{Pr} [ v_{c} \neq \hat{v}_{c}  ] $ in \eqref{eq:rate34647ee} is the error probability of the estimation of $v_{c}$ from $y_1$. In this case, it can be proved that $u_c  $ and $v_{c}$ can be estimated from $y_1$ in a successive way and the error probability of this estimation  vanishes as $P\to \infty$. The proof of this step is similar to that of Lemma~\ref{lm:errorcase1}, and hence it is omitted here to avoid the repetition.    
 Then, we have 
  \begin{align}
\Imu(v; y_1) &  \geq      \frac{   2- \alpha - \epsilon}{2} \log P + o(\log P). 
   \label{eq:rate4367}     
 \end{align} 
 As in the previous cases,   it can be proved that  $ \Imu ( v; y_{2} )   \leq    o(\log P) $.   
Finally we have a lower bound on the secure rate  $R  \geq  \frac{   2- \alpha - \epsilon}{2} \log P + o(\log P)  $, which implies that  the secure GDoF  $d =  2- \alpha  $ is achievable in this case.

\section{Proof of Lemma~\ref{lm:rateerror561}  \label{sec:epcase4} }

Given the observation $y_1$ in  \eqref{eq:yvk19637374}, and with  $5/6 \leq \alpha \leq 1$,  we will show that  $v_{c}, u_c, v_m$ and $v_{p}$ can be estimated with vanishing error probability, for almost all the channel realizations. 
Our focus is to  prove that $v_{c}, u_c$ and $v_m$ can be estimated from $y_1$ simultaneously with vanishing error probability, for almost all the channel realizations. 
This proof is motivated by the proof of  \cite[Lemma~4]{ChenIC:18}, in which the noise removal and signal separation techniques will be used.  Once $v_{c}, u_c$ and $v_m$ are  estimated correctly from $y_1$, we can remove $v_{c}, u_c$ and $v_m$ from $y_1$ and then estimate $v_{p}$ with vanishing error probability.

Recall that  $v_{c}, u_c \in    \Omega (\xi   =   \frac{ 6 \gamma}{Q},   \   Q =  P^{ \frac{ \alpha  -1/2 - \epsilon}{2}} ) $,  $v_m, u_{p}  \in    \Omega (\xi   = \frac{2 \gamma}{Q},   \   Q = P^{ \frac{ 1 - \alpha - \epsilon}{2}} ) $ and  $v_{p}  \in    \Omega (\xi   = \frac{ \gamma}{Q},   \   Q = P^{ \frac{ 1 - \alpha - \epsilon}{2}} ) $,  for some parameters $\gamma \in (0, 1/20]$ and $\epsilon \to 0$.
Let us describe  $y_1$ in the following form
   \begin{align}
y_{1} &=    \sqrt{P} h_{11}    h_{22} v_{c} +     \sqrt{P^{ 2 \alpha - 1  }}  h_{12}  h_{21}  u_c  +   \sqrt{P^{ 2 -2 \alpha }} h_{11}    h_{22} v_{m}  +    \sqrt{P^{ 1 - \alpha }} h_{11}    h_{22}  v_{p}     +      h_{12}  h_{21}  u_p   +  z_{1}    \non \\
&= \sqrt{P^{ 1 - \alpha + \epsilon }} \cdot 2\gamma  \underbrace{(    3 \sqrt{P^{ 1/2  }}   g_2 q_2 +  3 \sqrt{P^{ 2\alpha - 3/2  }}    g_1 q_1 +  g_0 q_0 )}_{ \defeq \tilde{s}  }   +  \sqrt{P^{ 1 - \alpha}}  \bigl(\underbrace{h_{11}    h_{22}  v_{p}  +    \sqrt{P^{ \alpha-1}}  h_{12}  h_{21}  u_p   \bigr)}_{\defeq  \tilde{e} }   +  z_{1}     \non\\
 &=        \sqrt{P^{ 1 - \alpha + \epsilon }} \cdot 2\gamma   \tilde{s}   +  \sqrt{P^{ 1 - \alpha}} \tilde{e}  +  z_{1}     \label{eq:y1rw2352738}  
\end{align}
where $g_2\defeq g_0 \defeq  h_{11}    h_{22} $,     $ g_1\defeq    h_{12}  h_{21}$,     $\tilde{e}   \defeq   h_{11}    h_{22}  v_{p}  +    \frac{1}{\sqrt{P^{ 1 - \alpha}}}  h_{12}  h_{21}  u_p  $,   $\tilde{s} \defeq g_0 q_0 + 3 \sqrt{P^{ 2\alpha - 3/2  }}   g_1 q_1 + 3 \sqrt{P^{ 1/2  }}  g_2 q_2 $ and  
    \[ q_2  \defeq  \frac{Q_{2}}{6\gamma} \cdot   v_{c}  ,    \quad  q_1  \defeq  \frac{Q_{1}}{6\gamma} \cdot   u_c, \quad   q_0  \defeq  \frac{ Q_{0}}{2\gamma} \cdot  v_{m}, \quad Q_{2} \defeq Q_{1}\defeq  P^{ \frac{ \alpha  -1/2 - \epsilon}{2}},\quad  Q_{0} \defeq   P^{ \frac{ 1 - \alpha  - \epsilon}{2}}.   \]
In this scenario,  the following conditions are always satisfied: $q_k \in \Zc$ and $|q_k| \leq Q_{k}$ for $k=0,1,2$. 
Let   \[ A_2  \defeq  3 \sqrt{P^{ 1/2  }} ,    \quad  A_1  \defeq 3 \sqrt{P^{ 2\alpha - 3/2  }}, \quad   A_0  \defeq  1.   \]
In this scenario with  $5/6 \leq \alpha \leq 1$, without loss of generality we will consider the  case  that\footnote{The result of Lemma~\ref{lm:rateerror561} still holds for the case when any of $\{Q_{0}, Q_{1}, Q_{2}, A_{1}, A_{2}\}$ is not integer. The proof just needs some minor modifications. For example, when $A_2$ is not an integer, we can modify  $ v_{c}$ and  $ u_{c}$ as $ v_{c}= \eta_c v'_{c}$ and $ u_{c}= \eta_c u'_{c}$, where $v'_{c}, u'_{c} \in    \Omega ( \xi =  \frac{  6 \gamma}{Q} ,    Q =  P^{ \frac{ \lambda_{c} }{2}} ) $, and  $\eta_c$ is a selected parameter such that  $0< \eta_c < 1$ and $A_2 \eta_c$ is an integer.  
}  
$Q_{0}, Q_{1}, Q_{2}, A_{1}, A_{2}  \in \Zc^+$.

For the observation $y_1$ in \eqref{eq:y1rw2352738}, the goal is to estimate the sum $\tilde{s} = g_0 (  q_0 + 3 \sqrt{P^{ 1/2  }} q_2) + 3 \sqrt{P^{ 2\alpha - 3/2  }}   g_1 q_1 $  by considering  the other signals as noise (noise removal).  After decoding  $\tilde{s} $ correctly, the three symbols  $q_0, q_1, q_2$ can be estimated, based on  the fact that $\{g_0, g_1\}$ are rationally independent (signal separation, cf.~\cite{MGMK:14}), as well as the fact that $q_0$ and $q_2$ can be reconstructed from $q_0 + 3 \sqrt{P^{ 1/2  }}  q_2$. Note that the minimum distance of $3 \sqrt{P^{ 1/2  }}  q_2$, i.e., $  \min_{  q_2, \bar{q}_2 \in \Zc  \cap [- Q_{2},    Q_{2}], \  q_2 \neq  \bar{q}_2    } 3 \sqrt{P^{ 1/2  }}  \cdot  |q_2 - \bar{q}_2   |$,  is no less  than the maximum of $2q_0$. 
To estimate  $\tilde{s}$ from $y_1$, we will show that the minimum distance of $\tilde{s}$ is sufficiently large, in order to make the error probability vanishing. 
Let us define the minimum distance of $\tilde{s}$  as
  \begin{align}
d_{\min}  (g_0, g_1, g_2)   \defeq    \min_{\substack{ q_1, q_2, \tilde{q}_1, \tilde{q}_2 \in \Zc  \cap [- Q_{2},    Q_{2}]  \\  q_0, \tilde{q}_0 \in \Zc  \cap [- Q_{0},    Q_{0}]   \\  (q_0, q_1, q_2) \neq  (\tilde{q}_0, \tilde{q}_1, \tilde{q}_2)  }}  | g_0  (q_0 - \tilde{q}_0) + 3 \sqrt{P^{ 2\alpha - 3/2  }} g_1 (q_1 - \tilde{q}_1) + 3 \sqrt{P^{ 1/2  }} g_2 (q_2 - \tilde{q}_2)  | .    \label{eq:md295267}
 \end{align}
 The following lemma provides a result on the minimum distance.

\begin{lemma}  \label{lm:dis8888}
For the case with $5/6 \leq \alpha \leq 1$, and  for some constants $\delta \in (0, 1]$ and  $\epsilon >0$,  the following bound on the minimum distance $d_{\min}$   holds true 
 \begin{align}
d_{\min}    \geq  \delta   \label{eq:dG666}
 \end{align}
for all  the channel  realizations $\{h_{11}, h_{12}, h_{22}, h_{21}\} \in (1, 2]^{2\times 2} \setminus \Ho$, where   $\Ho \subseteq (1,2]^{2\times 2}$ is an outage set  whose Lebesgue measure, denoted by $\mathcal{L}(\Ho)$,  has the following bound  
 \begin{align}
\mathcal{L}(\Ho) \leq    12096 \delta   \cdot     P^{ - \frac{ \epsilon  }{2}} .    \label{eq:LM777}
 \end{align}
 \end{lemma}
 \begin{proof}
 For $\beta \defeq  \delta  \in (0, 1]  $,  we define an event as
\begin{align}
B( q_2, q_1, q_0)  \defeq \{ (g_2,  g_1, g_0)  \in (1,4]^3 :  |  A_2 g_2 q_2  + A_1 g_1 q_1 +  g_0q_0   | < \beta \}     .  \label{eq:BO999}
\end{align}
Also define 
\begin{align}
B  \defeq   \bigcup_{\substack{ q_0, q_1, q_2 \in \Zc:  \\ |q_k| \leq  2 Q_k  \ \forall k  \\    (q_0, q_1, q_2) \neq  0}}  B(q_2, q_1, q_0) .   \label{eq:Boutage22}
\end{align}
For $5/6 \leq \alpha \leq 1$, by  \cite[Lemma~14]{NM:13} we have the following bound on the Lebesgue measure of $B$ (i.e., $\Lc (B )$)
\begin{align}
\Lc (B ) & \leq   504 \beta \cdot 4 \Bigl(  2 \min \bigl\{  \frac{  Q_{0}}{ A_2},  Q_{2} \bigr\}  +  \tilde{Q}_2 \cdot \min \bigl\{ Q_{1} ,  \frac{Q_{0} }{A_1},  \frac{  A_2 \tilde{Q}_2}{ A_1}   \bigr\} \non\\ &  \   \   \quad\quad + 2 \min \bigl\{    \frac{  Q_{0}}{ A_1} , Q_{1} \bigr\}  +  \tilde{Q}_1 \cdot \min \bigl\{ Q_{2} ,  \frac{Q_{0} }{A_2},  \frac{  A_1 \tilde{Q}_1}{ A_2}   \bigr\}     \Bigr)   \non\\
&\leq     504 \beta \cdot 4 \Bigl(   \frac{  2Q_{0}}{ A_2}   +  \tilde{Q}_2 \cdot   \frac{Q_{0} }{A_1}  +  \frac{  2Q_{0}}{ A_1}  +  \tilde{Q}_1 \cdot \frac{Q_{0} }{A_2}   \Bigr)   \non\\
&\leq     504 \beta \cdot 4 \Bigl(     Q_{1}      \cdot   \frac{9Q_{0} }{A_2}  +  \frac{  4Q_{0}}{ A_1}   \Bigr)   \non\\
&\leq     504 \beta \cdot 8 Q_{0}  \max \bigl\{        \frac{9 Q_{1} }{A_2},  \frac{  4}{ A_1}  \bigr\}   \non\\
&\leq     504 \beta \cdot 8 Q_{0}  \cdot 3   P^{ \frac{ \alpha  -1 }{2}}   \non\\
                & =  12096 \delta   \cdot     P^{ - \frac{ \epsilon  }{2}}          \label{eq:LB55555}
\end{align}
where   $\tilde{Q}_1 = \min\bigl\{ Q_{1},  8  \frac{\max\{ Q_{0},  A_2 Q_{2} \}}{A_1}\bigr\}  = Q_{1}$ and  $\tilde{Q}_2 = \min\bigl\{ Q_{2},    8  \frac{\max\{ Q_{0},   A_1 Q_{1} \}}{A_2}\bigr\}  = Q_{1}  \cdot \min\bigl\{ 1,    \frac{ 8A_1  }{A_2}\bigr\}$. 
In this scenario, we can treat $B$  as an outage set. When $(g_0, g_1, g_2)\notin B$, by definition  we have $d_{\min} (g_0, g_1, g_2)   \geq  \delta$. 
Recall that $g_2\defeq g_0 \defeq  h_{11}    h_{22} $ and  $ g_1\defeq    h_{12}  h_{21}$. 
At this point, we define a new set  $\Ho$  as
\[ \Ho \defeq \{  (h_{11}, h_{22}, h_{12}, h_{21} ) \in (1, 2]^{2\times 2} :      (  g_2=g_0, g_1, g_0) \in B  \} . \]
We define  $ \mathbbm{1}_{\Ho} (h_{11}, h_{22}, h_{12}, h_{21})= 1 $ if  $(h_{11}, h_{22}, h_{12}, h_{21}) \in  \Ho$, else   $\mathbbm{1}_{\Ho} (h_{11}, h_{22}, h_{12}, h_{21})= 0$. Similarly, we define $ \mathbbm{1}_{B}  ( g_1, g_0) =1$ if $(g_2 =g_0, g_1, g_0) \in  B$, else $ \mathbbm{1}_{B}  ( g_1, g_0) =0$. Then we can  bound the Lebesgue measure of $\Ho$ as 
\begin{align}
\Lc (\Ho ) & =  \int_{h_{11}=1}^2  \int_{h_{12}=1}^2 \int_{h_{21}=1}^2 \int_{h_{22}=1}^2    \mathbbm{1}_{\Ho}  (h_{11}, h_{22}, h_{12}, h_{21}) d h_{22} d h_{21} d h_{12}  d h_{11}          \non\\  
& =  \int_{h_{11}=1}^2  \int_{h_{12}=1}^2 \int_{h_{21}=1}^2 \int_{h_{22}=1}^2    \mathbbm{1}_{B}  (  h_{12}  h_{21} ,  h_{11}    h_{22}) d h_{22} d h_{21} d h_{12}  d h_{11}          \non\\ 
& \leq    \int_{h_{11}=1}^2   \int_{h_{12}=1}^2  \int_{g_{1} =1}^4 \int_{g_{0}=1}^4    \mathbbm{1}_{B}  ( g_{1}, g_{0})   h_{11}^{-1}  h_{12}^{-1}d g_{0} d g_{1} d h_{12}  d h_{11}          \non\\ 
& \leq    \int_{h_{11}=1}^2  \int_{h_{12}=1}^2   \mathcal{L}(B) d h_{12}  d h_{11}          \non\\ 
& \leq      12096 \delta   \cdot     P^{ - \frac{ \epsilon  }{2}}              \label{eq:LB3333}  
\end{align}
where the last step uses the result in \eqref{eq:LB55555}.
 \end{proof}

In the rest of this section, we will consider the channel  realizations $(h_{11}, h_{22}, h_{12}, h_{21}) \in (1, 2]^{2\times 2} $ that are not in the outage set   $\Ho$. The result of Lemma~\ref{lm:dis8888} reveals that  \[ \mathcal{L}(\Ho)  \to 0, \quad \text{for} \quad  P\to \infty. \] 
When the channel realizations satisfy the condition  $(h_{11}, h_{22}, h_{12}, h_{21}) \notin \Ho$, we have the following property on the minimum distance defined in \eqref{eq:md295267}:  $d_{\min}    \geq   \delta $ for a given constant $\delta \in (0, 1]$. 
With this result, we can estimate $\tilde{s}$ from $y_1$ expressed in \eqref{eq:y1rw2352738}.
For the random variable $\tilde{e}   =   h_{11}    h_{22}  v_{p}  +    \frac{1}{\sqrt{P^{ 1 - \alpha}}}  h_{12}  h_{21}  u_p  $ appeared in \eqref{eq:y1rw2352738}, it is true that 
\[ |\tilde{e} | \leq  \tilde{e}_{\max} \defeq  3/5  \quad  \forall   \tilde{e}.  \]  
At this point,  we have the following bound on the error probability of the estimation of $\tilde{s}$ from $y_1$ 
 \begin{align}
 \text{Pr} [ \tilde{s} \neq   \hat{\tilde{s}} ]    
  &\leq     \text{Pr} \Bigl[   | z_1  +    \sqrt{P^{ 1 - \alpha}} \tilde{e}  |  >   \sqrt{P^{ 1 - \alpha + \epsilon }} \cdot 2\gamma  \cdot \frac{d_{\text{min}} }{2}  \Bigr]   \non \\
       & \leq   2  \cdot     {\bf{Q}} \bigl(   P^{\frac{ 1 - \alpha + \epsilon}{2} } \cdot 2\gamma    \cdot \frac{d_{\text{min}} }{2}   -  P^{ \frac{1 - \alpha}{2}} \tilde{e}_{\max}  \bigr)    \non   \\ 
    & \leq    2  \cdot     {\bf{Q}} \bigl(  P^{ \frac{1 - \alpha}{2}} ( \gamma \delta  P^{ \frac{ \epsilon}{2}}   -  3/5)  \bigr)     \label{eq:error9982cc} 
  \end{align}
  where $\hat{\tilde{s}}$ denotes the estimate of $\tilde{s}$;
 ${\bf{Q}}(\tau )  \defeq  \frac{1}{\sqrt{2\pi}} \int_{\tau}^{\infty}  \exp( -\frac{ z^2}{2} ) d z$;
 the last step stems from the result that $d_{\min}    \geq   \delta $.
By following the fact that $ {\bf{Q}} (\tau ) \leq   \frac{1}{2}\exp ( - \tau^2 /2 ),  \    \forall \tau \geq 0$,  the result in \eqref{eq:error9982cc} implies the following conclusion 
 \begin{align}
 \text{Pr} [ \tilde{s} \neq   \hat{\tilde{s}} ]    \to 0  \quad  \text{for} \quad   P\to \infty.    \label{eq:eb7272}                           
 \end{align}
   After decoding  $\tilde{s} = g_0 (  q_0 + 3 \sqrt{P^{ 1/2  }} q_2) + 3 \sqrt{P^{ 2\alpha - 3/2  }}   g_1 q_1 $ correctly, the three symbols  $q_0, q_1, q_2$ can be recovered, as illustrated before in this section.

Next,  we  remove $\tilde{s}$ from $y_1$, which leads to
\begin{align}
 y_1  -  \sqrt{P^{ 1 - \alpha + \epsilon }} \cdot 2\gamma   \tilde{s} &=   \sqrt{P^{ 1 - \alpha}}  h_{11}    h_{22}  v_{p}  +      h_{12}  h_{21}  u_p      +  z_{1}      .     \label{eq:eb28267}  
\end{align}
Since the interference term $h_{12}  h_{21}  u_p$ in  \eqref{eq:eb28267} is under the noise level, i.e.,  $h_{12}  h_{21}  u_p  \leq 8\gamma \leq 2/5$, one can easily prove that the  
the error probability for decoding $v_{p}$  from the observation in  \eqref{eq:eb28267} is 
 \begin{align}
  \text{Pr} [ v_{p} \neq \hat{v}_{p} ]     \to 0  \quad  \text{for} \quad   P\to \infty . \label{eq:eb33334}                    
   \end{align}
  
Therefore, the error probability $\text{Pr} [ \{ v_{c} \neq \hat{v}_{c} \} \cup  \{ v_{m} \neq \hat{v}_{m} \} \cup  \{ v_{p} \neq \hat{v}_{p} \}  ]$ vanishes as $P$ approaches infinity, 
    for almost all the channel realizations  $(h_{11}, h_{22}, h_{12}, h_{21}) \in (1, 2]^{2\times 2} $.

\section{Conclusion}   \label{sec:concl333}

In this work,  we characterized the optimal secure GDoF  of a symmetric Gaussian  wiretap channel with a helper, under a weak notion of secrecy constraint. The result reveals that, adding a helper can significantly increase the secure GDoF of the wiretap channel.
A new converse and a new scheme are provided in this work.  The  converse derived in this work holds for the symmetric and asymmetric channels. 
In the proposed scheme, the helper sends a cooperative jamming signal at a specific power level and direction, which allows to  minimize the penalty in GDoF incurred by  the secrecy constraint.  
In the secure rate analysis, the techniques of noise removal and signal separation  are used.  
The optimal secure GDoF is described in  different expressions for different interference regimes. For the regimes of $0\leq \alpha \leq 5/6$ and $4/3 \leq \alpha \leq 2$, the achievable secure GDoF result holds for all the channel realizations under our channel model.  For the regime of $ 5/6< \alpha <4/3$, the achievable secure GDoF result holds for almost all the channel realizations when  $P$ is large, under our channel model. 
In the future work, we will generalize our secure GDoF result  to understand the constant-gap secure  capacity.

\appendices

\section{Proof of Lemma~\ref{lm:bound1112}} \label{sec:bound1112}

Recall that $\bar{y}_{2}(t)$ is a nosy version of $y_{2}(t)$,  defined in  \eqref{eq:defybdet}. From chain rule, we have
\begin{align}
   \Imu(w;   \bar{y}_{2}^{\bln})  &=     \Imu( w;  \bar{y}_{2}^{\bln} | s_{22}^{\bln})  +  \Imu( w; s_{22}^{\bln})   -   \Imu( w;  s_{22}^{\bln} |  \bar{y}_{2}^{\bln})     \non\\
   &=    \Imu( w;  \bar{y}_{2}^{\bln} | s_{22}^{\bln})    -   \Imu( w;  s_{22}^{\bln} |  \bar{y}_{2}^{\bln})   \label{eq:ups9883} 
\end{align}
where 
 \eqref{eq:ups9883} follows from the independence between $w$ and $s_{22}^{\bln}$. 
For the term $\Imu( w;  s_{22}^{\bln} |  \bar{y}_{2}^{\bln})$ in \eqref{eq:ups9883}, it can be bounded by 
\begin{align}
  &\Imu( w;  s_{22}^{\bln} |  \bar{y}_{2}^{\bln}) \non\\
    = &  \hen( s_{22}^{\bln} |  \bar{y}_{2}^{\bln})  -  \hen(   s_{22}^{\bln} |  \bar{y}_{2}^{\bln}, w)            \non\\
   \leq & \sum_{t=1}^{\bln}   \hen( s_{22} (t) |  \bar{y}_{2}(t))  -  \hen(   s_{22}^{\bln} |  \bar{y}_{2}^{\bln}, w,  x_{2}^{\bln})       \label{eq:ups35467}     \\
  = &\sum_{t=1}^{\bln}  \hen(   s_{22} (t)  -  \sqrt{P^{ - \phi_3 }}  \bar{y}_{2}(t)   |   \bar{y}_{2}(t))  - \underbrace{ \hen(   \tilde{z}_{2}^{\bln} ) }_{= \frac{\bln}{2}  \log (2\pi e)}     \non    \\
= &\sum_{t=1}^{\bln}  \hen \bigl(   \tilde{z}_{2}(t)    -   h_{21} x_{1}(t)    -  \sqrt{P^{ -  \alpha_{2 1} }}   h_{22} z_{2}(t)      -    \sqrt{P^{ - \phi_3 }} \bar{z}_{2}(t)    \non\\&  \quad\quad    +   ( \sqrt{P^{  (\alpha_{2 2} -  \alpha_{2 1})^+ }}  -  \sqrt{P^{  \alpha_{2 2} -  \alpha_{2 1} }} )   h_{22} x_{2}(t) 
 |   \bar{y}_{2}(t) \bigr)    -  \frac{\bln}{2}  \log (2\pi e) \label{eq:ups2457}     \\
\leq &    \frac{\bln}{2}  \log \bigl(  \underbrace{1 +  | h_{21}|^2  + P^{ -  \alpha_{2 1} }   | h_{22}|^2  +  P^{ - \phi_3 } }_{\leq 10}  +  \underbrace{( \sqrt{P^{  (\alpha_{2 2} -  \alpha_{2 1})^+ }}  -  \sqrt{P^{  \alpha_{2 2} -  \alpha_{2 1} }} )^2}_{\leq 1} \cdot  \underbrace{ | h_{22}|^2}_{\leq 4} \ \bigr)   \label{eq:ups22566}  \\
\leq &    \frac{\bln}{2}  \log  14   \label{eq:ups88753}      
\end{align}
where \eqref{eq:ups35467} follows from chain rule and  the fact that conditioning reduces differential entropy;
\eqref{eq:ups2457} uses the identity that $\hen(   \tilde{z}_{2}^{\bln} )  =\frac{\bln}{2}  \log (2\pi e)  $;
\eqref{eq:ups22566} results from the fact that Gaussian input maximizes the differential entropy, and that conditioning reduces differential entropy.
At this point, we complete the proof of Lemma~\ref{lm:bound1112}.

\section{Proof of Lemma~\ref{lm:differencefirst2}  } \label{sec:differencefirst2}

For  $s_{22}(t)$ and $\bar{y}_{2}(t)$ defined  in  \eqref{eq:defs11} and \eqref{eq:defybdet},  we have
 \begin{align}
 &     \hen( y^{\bln}_1|  s_{22}^{\bln})  - \hen( \bar{y}_{2}^{\bln} | s_{22}^{\bln})    \non\\
 \leq &\hen( y^{\bln}_1|  s_{22}^{\bln})  - \hen( \bar{y}_{2}^{\bln} | s_{22}^{\bln}, z_{2}^{\bln})  \label{eq:ups2566}     \\
  = &\hen( y^{\bln}_1|  s_{22}^{\bln})  - \hen( \{   \bar{y}_{2}(t)  -      \sqrt{P^{-(\alpha_{2 1} - \phi_3 )}}  z_{2}(t)        \}_{t=1}^{\bln} | s_{22}^{\bln}, z_{2}^{\bln})  \non    \\
    = &\hen( y^{\bln}_1|  s_{22}^{\bln})  - \hen( \{   \sqrt{P^{ \phi_3 }}  h_{21} x_{1}(t) +  \sqrt{P^{\alpha_{2 2} -(\alpha_{2 1} - \phi_3 )}}   h_{22} x_{2}(t)    +   \bar{z}_{2}(t)        \}_{t=1}^{\bln} | s_{22}^{\bln}, z_{2}^{\bln})  \non     \\
   = &\hen( y^{\bln}_1|  s_{22}^{\bln})  - \hen( \{   \sqrt{P^{ \phi_3 }}  h_{21} x_{1}(t) +  \sqrt{P^{\alpha_{2 2} -(\alpha_{2 1} - \phi_3 )}}   h_{22} x_{2}(t)    +   \bar{z}_{2}(t)        \}_{t=1}^{\bln} | s_{22}^{\bln})  \label{eq:ups8637}    \\ 
   = &\hen( y^{\bln}_1|  \{    \sqrt{P^{(\alpha_{22}-\alpha_{21})^+}} h_{22} x_{2}(t) +  z'_{2}(t)    \}_{t=1}^{\bln})    \non\\
    &- \hen( \{   \sqrt{P^{ \phi_3 }}  h_{21} x_{1}(t) +  \sqrt{P^{\alpha_{2 2} -(\alpha_{2 1} - \phi_3 )}}   h_{22} x_{2}(t)    +   \bar{z}_{2}(t)        \}_{t=1}^{\bln} |  \{    \sqrt{P^{(\alpha_{22}-\alpha_{21})^+}} h_{22} x_{2}(t) +  z'_{2}(t)    \}_{t=1}^{\bln})  \label{eq:ups826773}       
\end{align}
where 
 \eqref{eq:ups2566} uses the fact that conditioning reduces differential entropy;
 \eqref{eq:ups8637}  follows from the fact that $z_{2}^{\bln}$ is independent of $\{   \sqrt{P^{-  \phi_3 }}  h_{21} x_{1}(t) +  \sqrt{P^{\alpha_{2 2} -(\alpha_{2 1} - \phi_3 )}}   h_{22} x_{2}(t)    +   \bar{z}_{2}(t)        \}_{t=1}^{\bln}$ and $ s_{22}^{\bln} = \{    \sqrt{P^{(\alpha_{22}-\alpha_{21})^+}} h_{22} x_{2}(t) +  \tilde{z}_{2}(t)    \}_{t=1}^{\bln}$;
 in \eqref{eq:ups826773}  we replace $\tilde{z}_{2}(t) $ with a  new   noise random variable  $z'_{2}(t)  \sim \mathcal{N}(0, 1)$  that is  independent of the other noise random variables and transmitted signals $\{x_1(t), x_2 (t)\}_t$;
note that   replacing $\tilde{z}_{2}(t)  \sim \mathcal{N}(0, 1)$ with   $z'_{2}(t)  \sim \mathcal{N}(0, 1)$  will not change the differential entropies in \eqref{eq:ups8637}, due to the fact that differential entropy depends on  distributions.
To bound the right-hand side of \eqref{eq:ups826773}, we will use the  result of  \cite[Lemma~9]{ChenIC:18} that is described below.
\vspace{5pt}
\begin{lemma}  \cite[Lemma~9]{ChenIC:18}  \label{lm:differenceb}   
Let $y_{1}(t) = \sqrt{P^{\alpha_{11}}}  h_{11} x_{1}(t)+  \sqrt{P^{\alpha_{12}}}  h_{12} x_{2}(t) +z_{1}(t) $ and $y_{2}(t)=  \sqrt{P^{\alpha_{21}}}  h_{21} x_{1}(t)+  \sqrt{P^{\alpha_{22}}}  h_{22} x_{2}(t) +z_{2}(t) $, as defined  in \eqref{eq:Cmodel}. Consider a random variable (or a set of random variables),  $\bar{w}_1$, that is independent of $\{x^{\bln}_{2},   z_{1}^{\bln}, z_{2}^{\bln},  \tilde{z}_{1}^{\bln}, \tilde{z}_{2}^{\bln}, \bar{z}_{2}^{\bln}, \bar{z}_{3}^{\bln}, \bar{z}_{4}^{\bln}\}$;  and consider another  random variable (or another set of random variables), $\bar{w}_2$, that is independent of $\{x^{\bln}_{1},  z_{1}^{\bln}, z_{2}^{\bln},  \tilde{z}_{1}^{\bln}, \tilde{z}_{2}^{\bln}, \bar{z}_{2}^{\bln}, \bar{z}_{3}^{\bln}, \bar{z}_{4}^{\bln}\}$. Then,  we have 
\begin{align} 
   \hen(y^{\bln}_2| \bar{w}_1)  - \hen(y^{\bln}_1| \bar{w}_1)   \leq &  \frac{n}{2} \log \Bigl(1+ P^{\alpha_{22}-\alpha_{12}} \cdot \frac{|h_{22}|^2}{|h_{12}|^2} +  P^{\alpha_{21} -(\alpha_{11} - \alpha_{12} )^+} \cdot \frac{|h_{21}|^2}{|h_{11}|^2}  \Bigr) +    \frac{n}{2} \log 10   \label{eq:differenceb11G}  \\  
 \hen(y^{\bln}_1| \bar{w}_2)  - \hen(y^{\bln}_2| \bar{w}_2)  \leq  &   \frac{n}{2} \log \Bigl(1+ P^{\alpha_{11}-\alpha_{21}} \cdot \frac{|h_{11}|^2}{|h_{21}|^2} +  P^{\alpha_{12} -(\alpha_{22} - \alpha_{21} )^+} \cdot \frac{|h_{12}|^2}{|h_{22}|^2}  \Bigr) +    \frac{n}{2} \log 10 . \label{eq:differenceb22G}      
\end{align}
\end{lemma}
\vspace{5pt}
Note that the result in \eqref{eq:differenceb22G}  holds true when   $\bar{w}_2$ is set as $\bar{w}_2 \defeq \{    \sqrt{P^{(\alpha_{22}-\alpha_{21})^+}} h_{22} x_{2}(t) +  z'_{2}(t)    \}_{t=1}^{\bln}$.   Let us define  $\alpha_{21}' \defeq \phi_3 $,  $\alpha_{22}' \defeq \alpha_{2 2} -(\alpha_{2 1} - \phi_3 ) $ and  $ y_2 (t)'  \defeq  \sqrt{P^{ \alpha_{21}' }}  h_{21} x_{1}(t) +  \sqrt{P^{\alpha_{22}'}}   h_{22} x_{2}(t)    +   \bar{z}_{2}(t) $. Then, by incorporating the result of \eqref{eq:differenceb22G} into   \eqref{eq:ups826773},
we have  
 \begin{align}
    \hen( y^{\bln}_1|  s_{22}^{\bln})  - \hen( \bar{y}_{2}^{\bln} | s_{22}^{\bln})    \leq &  \hen(y^{\bln}_1| \bar{w}_2)  - \hen(y'^{\bln}_2| \bar{w}_2)  \label{eq:ups92866888}  \\
 \leq  &  \frac{n}{2} \log \Bigl(1+ P^{\alpha_{11}-\alpha_{21}'} \cdot \frac{|h_{11}|^2}{|h_{21}|^2} +  P^{\alpha_{12} -(\alpha_{22}' - \alpha_{21}' )^+} \cdot \frac{|h_{12}|^2}{|h_{22}|^2}  \Bigr) +    \frac{n}{2} \log 10    \label{eq:ups826773777}\\
 =  &   \frac{n}{2} \log \Bigl(1+ P^{\alpha_{11} - \phi_3} \cdot \frac{|h_{11}|^2}{|h_{21}|^2} +  P^{\alpha_{12} -(\alpha_{2 2} -\alpha_{2 1}  )^+} \cdot \frac{|h_{12}|^2}{|h_{22}|^2}  \Bigr)     +    \frac{n}{2} \log 10  \label{eq:ups92866}     
\end{align}
where \eqref{eq:ups92866888} is from \eqref{eq:ups826773};
and \eqref{eq:ups826773777} follows from  \eqref{eq:differenceb22G}.
Then, we complete the proof of Lemma~\ref{lm:differencefirst2}.

\section{Proof of Lemma~\ref{lm:differencea11}} \label{sec:differencea11}

In this section we provide the proof of Lemma~\ref{lm:differencea11}. Recall that  $y_{2}(t)$, $s_{22}(t) $, $\bar{x}_{1}(t)$ and $\bar{x}_{2}(t)$ are defined in  \eqref{eq:Cmodel},    \eqref{eq:defs11},  \eqref{eq:defx1b}, and  \eqref{eq:defx2b}, respectively.
In this setting, we have
   \begin{align}
& \quad \  \hen(     \bar{y}_{2}^{\bln} , s_{22}^{\bln} |  w)     -   \hen(  y^{\bln}_1, s_{22}^{\bln}| w)    \non\\
& =  \hen(     \bar{y}_{2}^{\bln} , s_{22}^{\bln} |  w)     -   \hen(  y^{\bln}_1,  s_{22}^{\bln}, \bar{x}_{1}^{\bln} | w)  + \underbrace{\hen(   \bar{x}_{1}^{\bln} | w, y^{\bln}_1,  s_{22}^{\bln}) }_{\defeq J_{11}}  \non\\
 & =  \hen(     \bar{y}_{2}^{\bln} ,     s_{22}^{\bln}, s_{12}^{\bln}, \bar{x}_{2}^{\bln}, \bar{x}_{1}^{\bln} |  w)  -  \underbrace{\hen(    s_{12}^{\bln}, \bar{x}_{2}^{\bln}, \bar{x}_{1}^{\bln} |  w,   \bar{y}_{2}^{\bln} ,     s_{22}^{\bln})      }_{\defeq J_{22}} -   \hen(  y^{\bln}_1, s_{22}^{\bln}, \bar{x}_{1}^{\bln}| w)  +  J_{11}  \non   \\
  & =  \hen(  s_{22}^{\bln}, s_{12}^{\bln}, \bar{x}_{2}^{\bln}, \bar{x}_{1}^{\bln} |  w)  +  \underbrace{ \hen(     \bar{y}_{2}^{\bln} | w,    s_{22}^{\bln}, s_{12}^{\bln}, \bar{x}_{2}^{\bln}, \bar{x}_{1}^{\bln} ) }_{\defeq J_{33} }     -   \hen(  y^{\bln}_1, s_{22}^{\bln}, \bar{x}_{1}^{\bln}| w)   +  J_{11}  - J_{22}   \non   \\
    & =    \hen(\bar{x}_{1}^{\bln} |  w)  + \hen(  s_{22}^{\bln}, s_{12}^{\bln} | \bar{x}_{1}^{\bln} , w)  +  \hen( \bar{x}_{2}^{\bln} |  s_{22}^{\bln}, s_{12}^{\bln}, \bar{x}_{1}^{\bln},  w)   \non\\& \quad   - \hen(  \bar{x}_{1}^{\bln}| w)  -   \hen(  y^{\bln}_1, s_{22}^{\bln} | \bar{x}_{1}^{\bln}, w)  +  J_{11} - J_{22}   +J_{33} \non   \\
    & =  \underbrace{   \hen(  s_{22}^{\bln}, s_{12}^{\bln} | \bar{x}_{1}^{\bln} , w)  }_{= \hen(  y^{\bln}_1, s_{22}^{\bln} | \bar{x}_{1}^{\bln}, w,   x_{1}^{\bln})}  -   \hen(  y^{\bln}_1, s_{22}^{\bln} | \bar{x}_{1}^{\bln}, w)  +  \hen( \bar{x}_{2}^{\bln} |  s_{22}^{\bln}, s_{12}^{\bln}, \bar{x}_{1}^{\bln},  w) +  J_{11}  - J_{22} +J_{33} \non   \\
   & =    \underbrace{\hen(  y^{\bln}_1, s_{22}^{\bln} | \bar{x}_{1}^{\bln}, w,   x_{1}^{\bln})  -   \hen(  y^{\bln}_1, s_{22}^{\bln} | \bar{x}_{1}^{\bln}, w)}_{ =    -  \Imu(  x_{1}^{\bln}; y^{\bln}_1, s_{22}^{\bln} | \bar{x}_{1}^{\bln}, w)  \leq 0 }  + \underbrace{ \hen( \bar{x}_{2}^{\bln} |  s_{22}^{\bln}, s_{12}^{\bln}, \bar{x}_{1}^{\bln},  w)}_{ \defeq  J_{44} }   +  J_{11}  - J_{22}  +J_{33}  \label{eq:ups8256}    \\
   & \leq        J_{11}  - J_{22} +J_{33}  + J_{44}.  \label{eq:ups9933}    
\end{align}
In the above steps, chain rules are used in the derivations. 
In addition,  \eqref{eq:ups9933} uses the fact that mutual information cannot be negative, and \eqref{eq:ups8256} follows  from the following derivations 
        \begin{align}
  \hen(  s_{22}^{\bln}, s_{12}^{\bln} | \bar{x}_{1}^{\bln} , w)    &=  \hen(  s_{22}^{\bln}, s_{12}^{\bln})     \non\\
  & =  \hen(  s_{22}^{\bln},  s_{12}^{\bln}  | \bar{x}_{1}^{\bln}, w,   x_{1}^{\bln})     \non\\
    & =  \hen(  s_{22}^{\bln},  \{ s_{12}(t)   +     \sqrt{P^{\alpha_{11}}}  h_{11} x_{1}(t)  \}_{t=1}^{\bln}  | \bar{x}_{1}^{\bln}, w,   x_{1}^{\bln})     \non\\
&=   \hen( s_{22}^{\bln}, y^{\bln}_1 | \bar{x}_{1}^{\bln}, w,   x_{1}^{\bln})   
\end{align}           
which use  the independence between $\{s_{22}^{\bln}, s_{12}^{\bln}\}$ and $\{\bar{x}_{1}^{\bln}, w,   x_{1}^{\bln}\}$, as well as the identity  $y_{1}(t) =     \sqrt{P^{\alpha_{11}}}  h_{11} x_{1}(t)+   s_{12}(t)  $.
To complete this proof, we invoke the following lemma. 
 \vspace{10pt}
\begin{lemma}  \label{lm:Jboundsnew}
For $J_{11}$, $J_{22}$, $J_{33}$, and $J_{44}$ defined in this section, we have
\begin{align}
J_{11}  & \leq   \frac{n}{2} \log (42 \pi e )           \label{eq:Jb1new} \\
J_{22}  & \geq   \frac{3n}{2}\log (2\pi e)    \label{eq:Jb2new}  \\
J_{33}  & \leq    \frac{n}{2} \log (16 \pi e )   \label{eq:Jb3new}  \\
J_{44}  & \leq      \frac{n}{2} \log \bigl(2\pi e\bigl(1+      P^{ \phi_3  - \phi_1  }  |h_{22}|^2      \bigr)\bigr)          \label{eq:Jb4new}  
\end{align}
where $\phi_3 \defeq \min\{\alpha_{21},  \alpha_{12},   (\alpha_{11} - \phi_1)^+   \}$  and  $\phi_1 \defeq  ( \alpha_{12}- (\alpha_{22} - \alpha_{21})^+)^+ $. 
\end{lemma}
The proof of Lemma~\ref{lm:Jboundsnew} is given in the following subsection.  
By incorporating the results of Lemma~\ref{lm:Jboundsnew} into  \eqref{eq:ups9933}, we have 
  \begin{align}
  \hen(     \bar{y}_{2}^{\bln} , s_{22}^{\bln} |  w)     -   \hen(  y^{\bln}_1, s_{22}^{\bln}| w)  
    &\leq        J_{11}  - J_{22} +J_{33}  + J_{44}  \non\\
   &\leq     \frac{n}{2} \log \bigl(1+      P^{ \phi_3  - \phi_1  }  |h_{22}|^2 \bigr)     +  \frac{n}{2} \log 168 
\end{align}
which completes the proof of Lemma~\ref{lm:differencea11}.

\subsection{Proof of Lemma~\ref{lm:Jboundsnew}} \label{app:Jboundsnew}

Recall that    $s_{11}(t)  =  \sqrt{P^{\alpha_{11}-\alpha_{12}}} h_{11} x_{1}(t) +  \tilde{z}_{1}(t)$,  $s_{22}(t)  =  \sqrt{P^{(\alpha_{22}-\alpha_{21})^+}} h_{22} x_{2}(t) +  \tilde{z}_{2}(t)$, $ s_{12}(t) =  \sqrt{P^{\alpha_{12}}} h_{12} x_{2}(t)  +z_{1}(t)$,  $\bar{x}_{1}(t)  \defeq  \sqrt{ P^{    \min\{\alpha_{21},  \alpha_{12},  \alpha_{11} -   \phi_1\}   }}  h_{21} x_{1}(t)  +   \bar{z}_{3}(t)  $, $ \bar{x}_{2}(t)  \defeq     \sqrt{P^{  \phi_3}}   \tilde{z}_{2}(t)   +   \bar{z}_{4}(t) $, $\bar{y}_{2}(t) \defeq  \sqrt{P^{-(\alpha_{2 1} - \phi_3 )}} y_{2}(t)    +   \bar{z}_{2}(t) $,  $\phi_3 \defeq \min\{\alpha_{21},  \alpha_{12},  \phi_2 \}$,     $ \phi_2 \defeq  (\alpha_{11} - \phi_1)^+$ and  $\phi_1 \defeq  ( \alpha_{12}- (\alpha_{22} - \alpha_{21})^+)^+ $.

At first we focus on the  bound of $J_{11}$:
\begin{align}
&J_{11}    \non\\
= &\hen(   \bar{x}_{1}^{\bln} | w, y^{\bln}_1,  s_{22}^{\bln})   \nonumber \\
  \leq & \sum_{t=1}^{\bln}  \hen( \bar{x}_{1} (t) |  y_1(t), s_{22}(t)  )  \label{eq:Jb9255}  \\
    = & \sum_{t=1}^{\bln}  \hen\bigl(   \bar{x}_{1} (t)   -       \sqrt{ P^{   \min\{\alpha_{21},  \alpha_{12},  \alpha_{11} -   \phi_1\}    -   \alpha_{11}  }}     \frac{h_{21} }{h_{11}}  \bigl( y_1(t)   -     \sqrt{P^{    \alpha_{12}   -    (\alpha_{22 } - \alpha_{21})^+  }}  \frac{h_{12} }{h_{22}}   s_{22}(t) \bigr)   |     y_1(t), s_{22}(t)  \bigr)   \label{eq:Jb2595}  \\
        = & \sum_{t=1}^{\bln}  \hen\bigl(     \bar{z}_{3}(t)   -     \sqrt{ P^{   \min\{\alpha_{21},  \alpha_{12},  \alpha_{11} -   \phi_1\}    -   \alpha_{11}  }}      \frac{h_{21} }{h_{11}} z_1(t)   \non\\& \quad  +    \sqrt{P^{  \min\{\alpha_{21},  \alpha_{12},  \alpha_{11} -   \phi_1\}    -   \alpha_{11}   +  \alpha_{12}   -    (\alpha_{22 } - \alpha_{21})^+  }}  \frac{h_{21}h_{12} }{h_{11} h_{22}}  \tilde{z}_{2}(t)    |     y_1(t), s_{22}(t)  \bigr)  \non  \\
       \leq&   \frac{n}{2} \log \bigl(2\pi e \bigl(1+  \underbrace{P^{   \min\{\alpha_{21},  \alpha_{12},  \alpha_{11} -   \phi_1\}    -   \alpha_{11}  }}_{\leq 1}  \cdot \underbrace{ \frac{|h_{21}|^2 }{|h_{11}|^2}}_{\leq 4}  +    \underbrace{ P^{  \min\{\alpha_{21},  \alpha_{12},  \alpha_{11} -   \phi_1\}    -   \alpha_{11}   +  \alpha_{12}   -    (\alpha_{22 } - \alpha_{21})^+  } }_{\leq 1} \cdot  \underbrace{\frac{|h_{21}|^2 |h_{12}|^2 }{|h_{11}|^2  |h_{22}|^2}     }_{\leq 16} \bigr)\bigr)  \label{eq:Jb44112} \\
       \leq &   \frac{n}{2} \log (42 \pi e )    \label{eq:Jb00998} 
\end{align}
where \eqref{eq:Jb9255}  follows from chain rule and the fact that conditioning reduces differential entropy;
  \eqref{eq:Jb2595} uses the fact that $\hen(a|b)=\hen(a-\beta b |b)$ for a constant $\beta$ and two continuous random variables $a$ and $b$; 
  \eqref{eq:Jb44112}  follows from  the fact that Gaussian input maximizes the differential entropy and that conditioning reduces differential entropy;
\eqref{eq:Jb00998} uses the identities $ \min\{\alpha_{21},  \alpha_{12},  \alpha_{11} -   \phi_1\}    -   \alpha_{11}  \leq 0 $ and  $   \min\{\alpha_{21},  \alpha_{12},  \alpha_{11} -   \phi_1\}    -   \alpha_{11}   +  \alpha_{12}   -    (\alpha_{22 } - \alpha_{21})^+   \leq 0 $, where $\phi_1 =  ( \alpha_{12}- (\alpha_{22} - \alpha_{21})^+)^+$.

For $J_{22}$, it can be bounded by 
\begin{align}
J_{22} & = \hen(    s_{12}^{\bln}, \bar{x}_{2}^{\bln}, \bar{x}_{1}^{\bln} |  w,   \bar{y}_{2}^{\bln} ,     s_{22}^{\bln})  \nonumber \\
       & \geq  \hen(    s_{12}^{\bln}, \bar{x}_{2}^{\bln}, \bar{x}_{1}^{\bln} |  w,   \bar{y}_{2}^{\bln} ,     s_{22}^{\bln}, x_{2}^{\bln}, x_{1}^{\bln},  \tilde{z}_{2}^{\bln})    \label{eq:Jb28586}  \\
     & = \hen(    z_{1}^{\bln},  \bar{z}_{4}^{\bln}, \bar{z}_{3}^{\bln}  )    \nonumber  \\
       & =  \frac{3n}{2}\log (2\pi e) 
\end{align}
where  \eqref{eq:Jb28586}  follows from the fact that conditioning reduces differential entropy.

For $J_{33}$, we have the following bound:
\begin{align}
J_{33} 
 = &  \hen(     \bar{y}_{2}^{\bln} | w,    s_{22}^{\bln}, s_{12}^{\bln}, \bar{x}_{2}^{\bln}, \bar{x}_{1}^{\bln} )   \nonumber \\
\leq &  \sum_{t=1}^\bln \hen(     \bar{y}_{2}(t) | s_{22}(t),  \bar{x}_{2} (t),  \bar{x}_{1}(t))   \label{eq:Jb2385677} \\
= &  \sum_{t=1}^\bln \hen(     \bar{y}_{2}(t)   - \bar{x}_{1} (t) -    \sqrt{P^{-(\alpha_{2 2}- \alpha_{2 1}  )^+ +  (\alpha_{2 2}- \alpha_{2 1}  )}}  \bigl(     \sqrt{P^{\phi_3 }}   s_{22}(t) -       \bar{x}_{2} (t) \bigr) | s_{22}(t), \bar{x}_{2} (t),  \bar{x}_{1}(t))  \non \\
= &  \sum_{t=1}^\bln \hen \bigl(  ( \sqrt{ P^{  \phi_3   }}  -  \sqrt{ P^{  \min\{\alpha_{21},  \alpha_{12},  \alpha_{11} -   \phi_1\}   }}  ) h_{21} x_{1}(t)  + \bar{z}_{2}(t)    +   \sqrt{ P^{  \phi_3 - \alpha_{21} }}      z_{2}(t) -   \bar{z}_{3}(t)   \non\\ & \quad  +   \sqrt{P^{-(\alpha_{2 2}- \alpha_{2 1}  )^+ +  (\alpha_{2 2}- \alpha_{2 1}  )}}    \bar{z}_{4} (t)   | s_{22}(t), \bar{x}_{2} (t),  \bar{x}_{1}(t) \bigr)  \non \\
  \leq&   \frac{n}{2} \log \bigl(2\pi e \bigl(  \underbrace{ ( \sqrt{ P^{  \phi_3   }}  -  \sqrt{ P^{   \min\{\alpha_{21},  \alpha_{12},  \alpha_{11} -   \phi_1\}    }}  )^2 }_{\leq 1}  |h_{21} |^2   + 1 + \underbrace{P^{  \phi_3 - \alpha_{21} }}_{\leq 1} +1 +    \underbrace{P^{-(\alpha_{2 2}- \alpha_{2 1}  )^+ +  (\alpha_{2 2}- \alpha_{2 1}  )}}_{\leq 1}   \bigr)\bigr)  \label{eq:Jb467747} \\
       \leq &   \frac{n}{2} \log (16 \pi e )    \label{eq:Jb9267336} 
\end{align}
where \eqref{eq:Jb2385677}  results from chain rule and the fact that conditioning reduces differential entropy;
  \eqref{eq:Jb467747}  follows from  the fact that Gaussian input maximizes the differential entropy and that conditioning reduces differential entropy;
  \eqref{eq:Jb9267336}  uses   the identity that $( \sqrt{ P^{   \min\{\alpha_{21},  \alpha_{12},  (\alpha_{11} -   \phi_1)^+ \}    }}  -  \sqrt{ P^{   \min\{\alpha_{21},  \alpha_{12},  \alpha_{11} -   \phi_1\}    }}  )^2 \leq 1$ and the definition that $\phi_3 \defeq \min\{\alpha_{21},  \alpha_{12},  \phi_2 \}$.

For the term $J_{44}$, we have two different bounds. One bound is given as 
\begin{align}
J_{44} =&  \hen( \bar{x}_{2}^{\bln} |  s_{22}^{\bln}, s_{12}^{\bln}, \bar{x}_{1}^{\bln},  w)    \nonumber \\
\leq &  \sum_{t=1}^\bln \hen(     \bar{x}_{2}(t) |   s_{22}(t), s_{12}(t))   \label{eq:Jb82667} \\
= &  \sum_{t=1}^\bln \hen\bigl(     \bar{x}_{2}(t)       -    \sqrt{P^{ \phi_3 }}  \bigl(  s_{22}(t) -   \sqrt{P^{  - \alpha_{12}  +  ( \alpha_{2 2} -\alpha_{2 1} )^+    }}   \frac{h_{22}}{h_{12}}  s_{12}(t)      \bigr)        |  s_{22}(t),  s_{12}(t) \bigr)  \non \\
= &  \sum_{t=1}^\bln \hen(     \bar{z}_{4}(t) +      \sqrt{P^{ \phi_3  - \alpha_{12}  +  ( \alpha_{2 2} -\alpha_{2 1} )^+    }}        \frac{h_{22}}{h_{12}}      z_{1}(t)  |  s_{22}(t),  s_{12}(t) )  \non \\
\leq &   \frac{n}{2} \log \Bigl(2\pi e\Bigl(1+ P^{ \phi_3  - (\alpha_{12}  -  ( \alpha_{2 2} -\alpha_{2 1} )^+)   } \cdot \frac{|h_{22}|^2}{|h_{12}|^2}  \Bigr)\Bigr)   \label{eq:Jb285773} 
\end{align}
where \eqref{eq:Jb82667}  results from chain rule and the fact that conditioning reduces differential entropy;
 \eqref{eq:Jb285773} follows from  the fact that Gaussian input maximizes the differential entropy and that conditioning reduces differential entropy.
 The other bound is given as
 \begin{align}
J_{44}  \leq &  \sum_{t=1}^\bln \hen(     \bar{x}_{2}(t) )   \label{eq:Jb28571} \\
= &   \frac{n}{2} \log \bigl(2\pi e\bigl(1+ P^{ \phi_3  }   \bigr)\bigr)   \label{eq:Jb9477331} 
\end{align}
 where \eqref{eq:Jb28571}  uses the fact that conditioning reduces differential entropy.
 By combining the bounds in \eqref{eq:Jb285773} and \eqref{eq:Jb9477331}, we finally have 
  \begin{align}
J_{44}  \leq  &    \frac{n}{2} \log \Bigl(2\pi e\Bigl(1+     \min \Bigl\{  P^{ \phi_3}, \   P^{ \phi_3  - (\alpha_{12}  -  ( \alpha_{2 2} -\alpha_{2 1} )^+)   } \cdot \frac{|h_{22}|^2}{|h_{12}|^2}    \Bigr\}    \Bigr)\Bigr)   \non\\
\leq   &   \frac{n}{2} \log \Bigl(2\pi e\Bigl(1+     \min \Bigl\{  P^{ \phi_3} |h_{22}|^2    , \   P^{ \phi_3  - (\alpha_{12}  -  ( \alpha_{2 2} -\alpha_{2 1} )^+)   }  |h_{22}|^2     \Bigr\}    \Bigr)\Bigr)     \non  \\
=   &   \frac{n}{2} \log \Bigl(2\pi e\Bigl(1+      P^{ \phi_3  - (\alpha_{12}  -  ( \alpha_{2 2} -\alpha_{2 1} )^+)^+   }  |h_{22}|^2      \Bigr)\Bigr)      \label{eq:Jb2866372} 
\end{align}
which completes the proof of Lemma~\ref{lm:Jboundsnew}.

\section{Proof of Lemma~\ref{lm:errorcase1}  \label{sec:errorcase1} }

Before showing the proof of Lemma~\ref{lm:errorcase1}, we describe the result of \cite[Lemma~1]{ChenIC:18} below, which will be used later. 
\begin{lemma}  \label{lm:icchenlemma} \cite[Lemma~1]{ChenIC:18}
Let $ y'=  \sqrt{P^{\alpha_1}}  h x + \sqrt{P^{\alpha_2}} e + z$, 
with three random variables $z\sim \mathcal{N}(0, \sigma^2)$,  $x \in \Omega (\xi,  Q)$,  and  $e  \in  \Sc_{e}$, for a given discrete set $\Sc_{e}$, under the condition of \[|e | \leq  e_{\max}, \quad \forall e \in  \Sc_{e}. \]  In this model,  $e_{\max}$, $h$, $\sigma$, $\alpha_1$ and $\alpha_2$   are positive constants  independent of $P$,  with a constraint that  $\alpha_1 > \alpha_2$. Let $\gamma' > 0$ be a finite constant independent of $P$.   If the parameters  $Q$ and $\xi$ are set as 
\begin{align}
Q =  \frac{P^{\frac{\alpha'}{2}} \cdot h \gamma' }{2 e_{\max} },    \quad  \quad    \xi =  \gamma' \cdot \frac{ 1}{Q},   \quad   \text{for} \quad    0<  \alpha' < \alpha_1 -\alpha_2    \non                            
\end{align}
then the error probability of the estimation of $x$ from $y'$ is \[ \text{Pr} (e) \to 0  \quad  \text{as} \quad   P\to \infty.\] 
\end{lemma}
\vspace{10pt}

Let us now prove Lemma~\ref{lm:errorcase1}.  Given the observation $y_1$ expressed in \eqref{eq:yvk12866ww}, we will show that  $v_{m}$ and $v_{p}$ can be estimated from $y_1$ with vanishing error probability. 
In this case,  $y_1$ can be described as 
 \begin{align}
y_{1} =    \sqrt{P} h_{11}    h_{22} v_{m}   +    \sqrt{P^{ 1 - \alpha }}  e'     +  z_{1}   \non
\end{align}
where $  e' \defeq   h_{11}    h_{22}  v_{p}    +     \sqrt{P^{ 3\alpha - 2 }}  h_{12}  h_{21}  u_p $.
In this scenario with $0\leq \alpha \leq 1/2$,  we have \[ |e' | \leq  7/5 \] for any realizations of $e'$.   Note that, 
  $v_{m}, u_p \in    \Omega (\xi   =   \frac{ 2 \gamma}{Q},   \   Q =  P^{ \frac{ \alpha  - \epsilon}{2}} ) $ and   $v_{p}  \in    \Omega (\xi   = \frac{ \gamma}{Q},   \   Q = P^{ \frac{ 1 - \alpha - \epsilon}{2}} ) $,  for some parameters $\gamma \in (0, 1/20]$ and $\epsilon \to 0$.  Then,  by Lemma~\ref{lm:icchenlemma}, it holds true that
the error probability of the estimation of  $v_{m}$ from $y_1$  is 
 \begin{align}
\text{Pr} [v_{m} \neq \hat{v}_{m}] \to 0, \quad \text {as}\quad  P\to \infty.   \label{eq:errorvc14262}
\end{align}

In the next step, we remove $v_{m}$ from $y_1$ and then estimate  $v_{p}$ from the following observation
 \begin{align}
y'_{1} =       \sqrt{P^{ 1 - \alpha }} h_{11}    h_{22}  v_{p}    +     \sqrt{P^{ 2 \alpha - 1   }}  h_{12}  h_{21}  u_p    +  z_{1} .   \label{eq:y11492626}
\end{align}
For the second term in the right-hand side of  \eqref{eq:y11492626}, the following condition is always satisfied
\[ | h_{12}  h_{21}  u_p | \leq  4 \times 2 \gamma \leq 2/5 .\]
Therefore, by  Lemma~\ref{lm:icchenlemma}, it is also true that
the error probability of the estimation of  $v_{p}$ from $y'_1$  expressed in  \eqref{eq:y11492626} is 
\begin{align}
\text{Pr} [v_{p}  \neq  \hat{v}_{p} |  v_{m} = \hat{v}_{m}]  \to 0, \quad \text {as}\quad  P\to \infty. \label{eq:errorvp02526}
\end{align}

At this point, by combining the results of  \eqref{eq:errorvc14262} and \eqref{eq:errorvp02526}, it gives 
\begin{align}
\text{Pr} [  \{ v_{m} \neq \hat{v}_{m} \} \cup  \{ v_{p} \neq \hat{v}_{p} \}  ]    \to   0         \quad \text {as}\quad  P\to \infty .  \non 
\end{align}

\section{Proof of Lemma~\ref{lm:epcase5}  \label{sec:epcase5} }

Given the case with  $1\leq  \alpha \leq  4/3$,  we will show that  $v_{c}$ and $u_c$  can be estimated from $y_1$  with vanishing error probability, for almost all the channel realizations. 
Recall that  $v_{c}, u_c \in    \Omega (\xi   =   \frac{ 6 \gamma}{Q},   \   Q =  P^{ \frac{ \alpha/2 - \epsilon}{2}} ) $,  for some parameters $\gamma \in (0, 1/20]$ and $\epsilon \to 0$.
Let us describe  $y_1$ in the following form
   \begin{align}
y_{1} &=   \sqrt{P^{ 2 - \alpha  }}   h_{11}    h_{22} v_{c}    +     \sqrt{P^{  \alpha  }}  h_{12}  h_{21}  u_c    +  z_{1}     \non \\
&= 6\gamma  P^{ \epsilon/2}  \underbrace{(A'_0 g'_0 q'_0 +  A'_1 g'_1 q'_1 )}_{\defeq \tilde{s}'}     +  z_{1}    \non\\
&=  6\gamma  P^{ \epsilon/2}   \tilde{s}'     +  z_{1}      \label{eq:y1998877}  
\end{align}
where $ g'_0 \defeq  h_{11}    h_{22} $,     $ g'_1\defeq    h_{12}  h_{21}$,  $A'_0  \defeq  \sqrt{P^{ 2 - 3\alpha/2  }}$, $A'_1  \defeq   \sqrt{P^{  \alpha /2  }} $,  $\tilde{s}' \defeq  A'_0 g'_0 q'_0 +  A'_1 g'_1 q'_1$  and  
    \[ q'_0  \defeq  \frac{Q'_{0}}{6\gamma} \cdot   v_{c}  ,    \quad  q'_1  \defeq  \frac{Q'_{1}}{6\gamma} \cdot   u_c, \quad Q'_{1}\defeq   Q'_{0} \defeq   P^{ \frac{ \alpha/2  - \epsilon}{2}}.   \]
In this scenario,  the following conditions are always satisfied: $q'_0, q'_1 \in \Zc$, $|q'_0| \leq Q'_{0}$ and $|q'_1| \leq Q'_{0}$. 
Similar to the proof of Lemma~\ref{lm:rateerror561}, without loss of generality we will consider the  case  that $Q'_{0}, A'_{0}, A'_{1}  \in \Zc^+$.

For the observation $y_1$ in \eqref{eq:y1998877}, our focus  is to estimate the sum $\tilde{s}' =  A'_0 g'_0 q'_0 +  A'_1 g'_1 q'_1$.  After decoding  $\tilde{s}' $ correctly, $q'_0$ and $q'_1$ can be recovered, because $\{g'_0, g'_1\}$ are rationally independent. 
We define the minimum distance of $\tilde{s}'$  as
  \begin{align}
d'_{\min}  (g'_0, g'_1)   \defeq    \min_{\substack{ q'_0, q'_1, \tilde{q}'_0, \tilde{q}'_1 \in \Zc  \cap [- Q'_{0},    Q'_{0}]   \\  (q'_0, q'_1) \neq  (\tilde{q}'_0, \tilde{q}'_1)  }}     |  A'_0  g'_0  (q'_0 - \tilde{q}'_0) + A'_1 g'_1 (q'_1 - \tilde{q}'_1)   | .    \label{eq:md223344}
 \end{align}
 The following lemma provides a result on the minimum distance.

\begin{lemma}  \label{lm:dis143999}
For the case with $1\leq  \alpha \leq  4/3$, and  for some constants $\delta \in (0, 1]$ and  $\epsilon >0$,  the following bound on the minimum distance $d'_{\min}$  defined in \eqref{eq:md223344}  holds true 
 \begin{align}
d'_{\min}    \geq  \delta  \label{eq:dG666888}
 \end{align}
for all  the channel  realizations $\{h_{11}, h_{12}, h_{22}, h_{21}\} \in (1, 2]^{2\times 2} \setminus \Hop$, where   $\Hop \subseteq (1,2]^{2\times 2}$ is an outage set  whose Lebesgue measure, denoted by $\mathcal{L}(\Hop)$,  has the following bound  
 \begin{align}
\mathcal{L}(\Hop) \leq    192 \delta   \cdot     P^{ - \frac{ \epsilon  }{2}} .    \label{eq:LM777888}
 \end{align}
 \end{lemma}
 \begin{proof}
For  $\beta \defeq   \delta \in (0, 1]  $,  we define an event as
\begin{align}
B'(  q'_1, q'_0)  \defeq \{ (g'_1, g'_0)  \in (1,4]^2 :  |  A'_1 g'_1 q'_1 +  A'_0 g'_0q'_0   | < \beta \}       \label{eq:BO999888}
\end{align}
and define 
\begin{align}
B'  \defeq   \bigcup_{\substack{ q'_0, q'_1 \in \Zc:  \\ |q'_k| \leq 2 Q'_0  \ \forall k  \\    (q'_0, q'_1) \neq  0}}  B'( q'_1, q'_0) .   \label{eq:Boutage2233}
\end{align}
For this case with $1\leq  \alpha \leq  4/3$, by  \cite[Lemma~1]{ChenLiArxiv:18} we have a bound on  the Lebesgue measure of $B'$, given as
\begin{align}
\Lc (B' ) & \leq   24 \beta    \min\Bigl\{ \frac{  4 Q'_1 Q'_0}{A'_1},  \frac{4 Q'_0 Q'_1}{A'_0},   \frac{ 8 Q'_0 }{A'_1},   \frac{8Q'_1 }{A'_0}\Bigr\}     \non\\
&\leq     24 \beta  \cdot  \frac{ 8 Q'_0 }{A'_1}   \non\\
 & =  192 \delta   \cdot     P^{ - \frac{ \epsilon  }{2}} .         \label{eq:LB55555666}
\end{align}
At this point, we define a new set  $\Hop$  as
\[ \Hop \defeq \{  (h_{11}, h_{22}, h_{12}, h_{21} ) \in (1, 2]^{2\times 2} :      (  g'_1, g'_0) \in B'  \} . \]
By following the steps related to \eqref{eq:LB3333}, we have the following bound on the Lebesgue measure of $\Hop$  
\begin{align}
\Lc (\Hop )  \leq    \Lc (B' ) \leq 192 \delta   \cdot     P^{ - \frac{ \epsilon  }{2}}     .        \label{eq:LB444333555}  
\end{align}
 \end{proof}

Lemma~\ref{lm:dis143999}  reveals  that  the Lebesgue measure of the outage set $\Hop$ is vanishing when $P$ is large, i.e.,  \[ \mathcal{L}(\Hop)  \to 0, \quad \text{for} \quad  P\to \infty. \]
Let us now  consider the channel  condition that  $(h_{11}, h_{22}, h_{12}, h_{21}) \notin \Hop$, in which the minimum distance of $\tilde{s}'$, defined in \eqref{eq:md295267}, satisfies the  inequality of $d'_{\min}    \geq   \delta $ (see \eqref{eq:dG666888}). 
With this result, we can  conclude that  the error probability for decoding $\tilde{s}'$ from $y_1 = 6\gamma  P^{ \epsilon/2} \tilde{s}'  +  z_{1} $ (see  \eqref{eq:y1998877}), denoted by $ \text{Pr} [ \tilde{s}' \neq   \hat{\tilde{s}} ] $,  is 
 \begin{align}
 \text{Pr} [ \tilde{s}' \neq   \hat{\tilde{s}}' ]    \to 0  \quad  \text{for} \quad   P\to \infty  \non                       
 \end{align}
 for almost all the channel realizations in the regime of large $P$.  
 After decoding  $\tilde{s}' $ correctly, $q'_0$ and $q'_1$ can be recovered, based on  the fact that $\{g'_0, g'_1\}$ are rationally independent.  Then, we complete the proof.

\section{Secure GDoF of the  Gaussian  wiretap channel \emph{without} a helper} \label{sec:GDoFnoHelper}

This section  focuses on the wiretap channel \emph{without} a helper (removing transmitter~2).
For this channel, the goal is to understand  the GDoF based on the capacity result of \cite{Wyner:75, CsiszarKorner:78}, which will be used for the GDoF comparison of the wiretap channels with and without a helper. 
For the wiretap channel \emph{without} a helper, the secure capacity, denoted by $C_{no}$, is given by:
\begin{align}
C_{no}=\max_{v \to x_1 \to y_1, y_2} \Imu(v; y_1) - \Imu(v; y_2)     \label{eq:GDoFnoH91566} 
\end{align}
(cf.~\cite{Wyner:75, CsiszarKorner:78}), where the maximum is computed over all random variables $v, x_1, y_1, y_2$ such that  $v \to x_1 \to y_1, y_2$ forms a Markov chain, and $y_{k}=   \sqrt{P^{\alpha_{k1}}} h_{k1} x_{1} +z_{k}$ for  $k=1,2$. 
Let us focus on the  upper bound on the following difference:
\begin{align}
&\Imu(v; y_1) - \Imu(v; y_2)  \non \\
\leq& \Imu(v; y_1, y_2) - \Imu(v; y_2)  \non \\
=& \hen(y_1| y_2)  - \hen(y_1| y_2,v)   \non  \\
\leq & \hen(y_1| y_2)  - \hen(y_1| y_2,v, x_1)   \label{eq:GDoFnoH33556}  \\
= & \hen(y_1| y_2)  - \frac{1}{2} \log (2\pi e)   \label{eq:GDoFnoH92455}  \\
= & \hen(y_1 -   \sqrt{P^{ \alpha_{11} - \alpha_{21} }} \frac{h_{11}}{h_{21}}  y_2  | y_2)  - \frac{1}{2} \log (2\pi e)   \non  \\
= & \hen(z_1 -   \sqrt{P^{ \alpha_{11} - \alpha_{21} }} \frac{h_{11}}{h_{21}}  z_2  | y_2)  - \frac{1}{2} \log (2\pi e)   \non  \\
\leq & \hen(z_1 -   \sqrt{P^{ \alpha_{11} - \alpha_{21} }} \frac{h_{11}}{h_{21}}  z_2  )  - \frac{1}{2} \log (2\pi e)   \label{eq:GDoFnoH7477}  \\
= &  \frac{1}{2} \log (1+  P^{ \alpha_{11} - \alpha_{21} }  \frac{|h_{11}|^2}{|h_{21}|^2} )  \label{eq:GDoFnoH8898} 
\end{align}
where  \eqref{eq:GDoFnoH33556}  and \eqref{eq:GDoFnoH7477} use the identity that conditioning reduces differential entropy;
 \eqref{eq:GDoFnoH92455}  results from the fact that $\hen(y_1| y_2,v, x_1) = \hen(z_1) = \frac{1}{2} \log (2\pi e)$.
 By combining  \eqref{eq:GDoFnoH91566}  and \eqref{eq:GDoFnoH8898}, the secure GDoF, denoted by $d_{no} $, is upper bounded by 
 \begin{align}
d_{no}  \leq  (\alpha_{11} - \alpha_{21} )^+ .   \label{eq:GDoFnoHupper} 
\end{align}
On the other hand, since the secure capacity is optimized over the random variables $v$ and $x_1$, by setting $x_1 = v \sim \mathcal{N}(0, 1)$ we have the lower bound on the secure capacity:
\begin{align}
C_{no} \geq &  \Imu(v; y_1) - \Imu(v; y_2)    \non   \\
= &   \hen(\sqrt{P^{\alpha_{11}}} h_{11} x_{1} +z_1)  - \hen(z_1)   - \hen(\sqrt{P^{\alpha_{21}}} h_{21} x_{1} +z_2)  + \hen(z_2)    \non   \\
= &  \frac{1}{2} \log (1+  P^{ \alpha_{11} }  |h_{11}|^2 )   -   \frac{1}{2} \log (1+  P^{ \alpha_{21} }  |h_{21}|^2 ) .  \label{eq:GDoFnoHlowerc} 
\end{align}
The bound in \eqref{eq:GDoFnoHlowerc} reveals that the secure GDoF is lower bounded by 
 \begin{align}
d_{no}  \geq  (\alpha_{11} - \alpha_{21} )^+ 
\end{align}
which, together with \eqref{eq:GDoFnoHupper}, gives the optimal secure GDoF 
 \begin{align}
d_{no}  =  (\alpha_{11} - \alpha_{21} )^+. 
\end{align}
For the symmetric case of notation with $\alpha_{11}=1$ and $\alpha_{12}= \alpha$, this  secure GDoF becomes
\begin{align}
d_{no}  = (1 -  \alpha)^+       \quad \forall  \alpha \in [0,  \infty) .  \non 
\end{align}

\section{Proof of Corollary~\ref{cor:symSGDoF}} \label{sec:symSGDoF}

For the symmetric setting with $\alpha_{11}= \alpha_{22}=1,   \alpha_{21}= \alpha_{12}=\alpha$, $ \phi_1$ and $ \phi_3$ take the following forms: 
\begin{align}
\phi_1  & = (\alpha - (1 - \alpha)^+)^+ \non \\
\phi_3 &=  \min\{ \alpha,  ( 1   - (\alpha - (1 - \alpha)^+)^+)^+ \}.   \non
\end{align}
In this symmetric case,  the three bounds  in Corollary~\ref {cor:SGDoF} then become
\begin{align}
 d   &\leq   \max \{  \phi_1    ,  (1 -  \phi_3)^+ \} +  (\phi_3   -    \phi_1 )^+  \label{eq:detupWTsym}   \\
 d   &\leq   (1 - \alpha )^+   +   \frac{\max\{1, \alpha  \}}{2}     \label{eq:detup3sym}     \\
    d   &\leq  (2 - \alpha)^+  .  \label{eq:detupsym1}   
\end{align}

When  $0 \leq \alpha \leq 1/2$, it reveals that  $\phi_1  = 0$ and  $\phi_3 =   \alpha $, and the bounds in \eqref{eq:detupWTsym}-\eqref{eq:detupsym1}  can be simplified as
\begin{align*}
d    &\leq     1    \\
  d  & \leq        3/2 - \alpha  \\
  d  &\leq   2 - \alpha  
\end{align*}
which implies that 
 \begin{align*}
d &\leq   \min\{1 ,  3/2 - \alpha, 2 - \alpha \}  = 1 ,  \quad \forall  \alpha\in [0, 1/2 ].
  \end{align*}
  
 When $1/2 \leq \alpha  \leq  1$,  it suggests that $\phi_1 =  2\alpha- 1 $ and $\phi_3 =   \min\{ \alpha, 2 ( 1   -  \alpha  ) \} $. Then,  the bounds in  \eqref{eq:detupWTsym}-\eqref{eq:detupsym1}  can be simplified as 
 \begin{align*}
d   &\leq   \max \{  2\alpha- 1    ,   1 - \alpha  \} +  \min\{ 1-  \alpha  ,  (3  - 4\alpha )^+ \}        \\
 d  & \leq   3/2 - \alpha   \\
 d  &\leq  2 - \alpha  . \\
\end{align*}
From the above results, the GDoF  $d$ can be bounded as 
 \begin{align*}
d &\leq  \min\{ 2- 2 \alpha, 3/2 - \alpha,  2 - \alpha  \}   = 2- 2 \alpha  , \  \quad \forall  \alpha\in [1/2 , 3/4]  \\
d &\leq   \min\{ 2\alpha- 1 , 3/2 - \alpha,  2 - \alpha  \}   =  2\alpha- 1,   \  \quad \forall  \alpha\in [3/4 , 5/6]  \\
d &\leq  \min\{ 2\alpha- 1 , 3/2 - \alpha,  2 - \alpha  \} = 3/2 - \alpha,  \quad \forall  \alpha\in [5/6 , 1] .
  \end{align*}
 When  $1 \leq \alpha $,   then $\phi_1 =  (\alpha - (1 - \alpha)^+)^+     = \alpha$ and  $\phi_3 =  0$,
  and  the bounds in \eqref{eq:detupWTsym}-\eqref{eq:detupsym1}  can be simplified as 
  \begin{align*}
 d   &\leq    \alpha    \\
 d   &\leq         \alpha/2         \\
    d   &\leq  (2 - \alpha)^+   .
\end{align*}
The above results imply that 
 \begin{align*}
 d &\leq  \min\{ (2 - \alpha)^+  , \   \alpha/2 \}  =\alpha/2, \quad \quad \forall  \alpha\in [1, 4/3] \\
 d &\leq  \min\{ (2 - \alpha)^+  , \  \alpha/2 \}  = 2 - \alpha,  \  \quad \forall  \alpha\in [4/3, 2]  \\
 d &\leq  \min\{ (2 - \alpha)^+  , \  \alpha/2 \}  = 0 ,  \  \quad \quad\quad \forall  \alpha\in [2, +\infty].  
  \end{align*}
  At this point we complete the proof.



\end{document}